\documentclass[a4paper,UKenglish, autoref, thm-restate]{lipics-v2021}
\usepackage[numbers]{natbib}

\usepackage{subcaption}
\usepackage{xcolor}
\usepackage{colortbl}
\usepackage{booktabs} 
\usepackage[export]{adjustbox}
\usepackage[ruled]{algorithm2e}

\SetAlFnt{\small}
\SetAlCapFnt{\small}
\SetAlCapNameFnt{\small}
\SetAlCapHSkip{0pt}
\IncMargin{-\parindent}
\usepackage[capitalize]{cleveref}
\newtheorem{property}{Property}
\crefname{property}{Property}{Properties}
\usepackage{enumitem}
\usepackage{amsthm}
\usepackage{amsmath}
\usepackage{thm-restate}
\definecolor{color1}{HTML}{05668D}
\definecolor{color2}{HTML}{DA3E52}
\definecolor{color3}{HTML}{F8BE57}
\definecolor{color4}{HTML}{F2BAC9}  
\definecolor{color5}{HTML}{3c096c}
\definecolor{color6}{HTML}{FF784F}
\definecolor{color7}{HTML}{02C39A}
\definecolor{color8}{HTML}{3c096c}  
\definecolor{twogreen}{HTML}{A7D7A7} 
\definecolor{twodarkgreen}{HTML}{C1E1C1}
\definecolor{twoyellow}{HTML}{FFDD94}
\definecolor{twored}{HTML}{FFADAD}
\usepackage{tikz}
\newcommand{\tikzxmark}{%
\tikz[scale=0.23] {
    \draw[line width=0.7,line cap=round] (0,0) to [bend left=6] (1,1);
    \draw[line width=0.7,line cap=round] (0.2,0.95) to [bend right=3] (0.8,0.05);
}}
\newcommand{\tikzcmark}{%
\tikz[scale=0.23] {
    \draw[line width=0.7,line cap=round] (0.25,0) to [bend left=10] (1,1);
    \draw[line width=0.8,line cap=round] (0,0.35) to [bend right=1] (0.23,0);
}}

\newcommand{\tx}{\textit{tx}}

\newcommand{\txtime}{t}
\newcommand{\txkeys}{K}

\newcommand{\txset}{T}

\newcommand{\TFM}{\textsf{TFM}\xspace}
\newcommand{\GCM}{\textsf{GCM}\xspace}
\newcommand{\GCMs}{\textsf{GCM}s\xspace}
\newcommand{\TFMs}{\textsf{TFM}s\xspace}

\newcommand{\gas}{\textsf{gas}}

\allowdisplaybreaks[1]

\hideLIPIcs  

\title{Transaction Fee Market Design for Parallel Execution} 

\author{Bahar Acilan\footnote{The authors of this work are listed alphabetically.}}{ETH Zurich}{bacilan@ethz.ch}{https://orcid.org/0009-0005-9241-0369}{}
\author{Andrei Constantinescu}{ETH Zurich}{aconstantine@ethz.ch}{https://orcid.org/0009-0005-1708-9376}{}
\author{Lioba Heimbach}{ETH Zurich}{hlioba@ethz.ch}{https://orcid.org/0000-0002-8258-1712}{}

\author{Roger Wattenhofer}{ETH Zurich}{wattenhofer@ethz.ch}{https://orcid.org/0000-0002-6339-3134}{}

\authorrunning{B. Acilan, A. Constantinescu, L. Heimbach and R. Wattenhofer}

\Copyright{CC-BY} 
\ccsdesc[500]{Applied computing~Economics}

\keywords{blockchain, transaction fee mechanism, parallel execution} 

\category{} 

\acknowledgements{We thank the Robust Incentives Group at the Ethereum Foundation for supporting this work through a grant. We thank Boyan Zhou for his work on a student project, supervised by the last three authors, which helped shape this work.}

\nolinenumbers 

\EventEditors{Zeta Avarikioti and Nicolas Christin}
\EventNoEds{2}
\EventLongTitle{7th Conference on Advances in Financial Technologies (AFT 2025)}
\EventShortTitle{AFT 2025}
\EventAcronym{AFT}
\EventYear{2025}
\EventDate{October 8--10, 2025}
\EventLocation{Pittsburgh, PA, USA}
\EventLogo{}
\SeriesVolume{354}
\ArticleNo{14}

\begin{document}

\maketitle

\begin{abstract}
Given the low throughput of blockchains like Bitcoin and Ethereum, scalability --- the ability to process an increasing number of transactions --- has become a central focus of blockchain research. One promising approach is the parallelization of transaction execution across multiple threads. However, achieving efficient parallelization requires a redesign of the incentive structure within the fee market. Currently, the fee market does not differentiate between transactions that access multiple high-demand storage keys (i.e., unique identifiers for individual data entries) versus a single low-demand one, as long as they require the same computational effort. Addressing this discrepancy is crucial for enabling more effective parallel execution.

In this work, we aim to bridge the gap between the current fee market and the need for parallel execution by exploring alternative fee market designs. To this end, we propose a framework consisting of two key components: a \textit{Gas Computation Mechanism (GCM)}, which quantifies the load a transaction places on the network in terms of parallelization and computation, measured in \textit{units of gas}, and a \textit{Transaction Fee Mechanism (TFM)}, which assigns a price to each unit of gas. We additionally introduce a set of desirable properties for a \GCM, propose several candidate mechanisms, and evaluate them against these criteria. Our analysis highlights two strong candidates: the \textit{weighted area} \GCM, which integrates smoothly with existing TFMs such as EIP‑1559 and satisfies a broad subset of the outlined properties, and the \textit{time-proportional makespan} \GCM, which assigns gas costs based on the context of the entire block’s schedule and, through this dependence on the overall execution outcome, captures the dynamics of parallel execution more accurately.
\end{abstract}

\section{Introduction}
Scalability, the ability to process more transactions efficiently, has become a central focus in blockchain research, especially given the low throughput of many existing networks. Ethereum, for example, is constrained by its single-threaded execution model, limiting transaction throughput. One promising way to enhance scalability is by parallelizing transaction execution across multiple threads, taking advantage of the multi-core processors common in modern hardware. However, achieving the full efficiency gains of parallel execution requires rethinking the fee market design to better account for \textit{storage key} contention and scheduling constraints. A storage key is a unique identifier (like an address label) for a specific data item in storage.

A \textit{transaction fee mechanism (TFM)} is a core component of any blockchain protocol, determining which pending transactions should be processed and what users must pay for the privilege of having their transactions executed. Traditional fee mechanisms, like Bitcoin’s first-price auction, involve users submitting bids with their transactions, and the transactions with the highest bids per computation are included in the next block. Ethereum initially used a similar purely first-price auction-based model but switched to the more sophisticated EIP-1559 mechanism in 2021, which introduces a fluctuating base fee based on network demand, aiming to improve incentive compatibility and reduce price volatility~\cite{eip1559spec}. Most of the existing literature on transaction fee mechanisms focuses on settings where transactions are executed sequentially and therefore does not account for storage key contention, which is crucial in parallel execution environments.

Thus, viewed in isolation, these traditional TFMs are ill-suited to the complexities introduced by parallel transaction execution. They price transactions solely based on their computational cost, without distinguishing between those that access a single storage key and those that interact with multiple, potentially contested storage keys. This pricing model works in a single-threaded environment but fails to capture the nuances of parallel execution, where transactions may impose vastly different constraints on storage key scheduling. Transactions that touch multiple high-contention storage keys can introduce bottlenecks, while those interacting with isolated storage keys are far easier to schedule efficiently.

Ethereum has yet to adopt parallel execution, but several blockchains, such as Solana, Aptos, and Sui, already employ parallel transaction processing~\cite{aptos,monad,sei_protocol,solana,sui}. However, many of these networks have yet to implement fee models that fully account for the challenges of parallel execution. Sui and Solana have introduced fee markets tailored to parallelization, but these mechanisms require users to engage in first-price auctions for congested storage keys~\cite{lostin2025truth,mueller2025}. As a result, these fee markets demand a high level of sophistication from users to effectively optimize their fee settings and are also not incentive-compatible. The requirements for an effective fee market that is suitable for parallel execution and the design of such a market have so far remained unresolved.

\paragraph*{Our Contributions.} In this work, we aim to bridge this gap by outlining the requirements for such a fee market and evaluating possible candidates. We outline our main contributions below: 
\begin{itemize}[topsep=0pt,itemsep=-1ex,partopsep=1ex,parsep=1ex]
    \item We introduce a framework with two main components: a \textit{Gas Computation Mechanism (GCM)}, which measures the load a transaction imposes on the network in terms of both parallel execution and computation, expressed in \emph{units of gas}, and a Transaction Fee Mechanism (TFM), which determines the cost associated with each unit of gas.
    \item  We introduce a list of desirable properties for a GCM and evaluate against them a set of mechanisms that we propose. 
    \item Our analysis identifies two promising candidates. The weighted area \GCM achieves a large subset of the outlined properties and, importantly, can be seamlessly integrated with existing TFMs, inheriting their properties. Complementing it, the time-proportional makespan \GCM is designed to price execution costs in proportion to the total computational load of a block, allowing more accurate resource pricing at the block level.
\end{itemize}

\section{Model}

\newcommand{\keyset}{\mathcal{K}}

We consider a universe consisting of several stateful \emph{storage keys} (e.g., user accounts, storage addresses of smart contracts). Each storage key can be thought of as a system global variable. We write $\keyset$ for the set of storage keys. For analysis purposes, we assume that $\keyset$ is infinite.
A \emph{transaction} is a sequence of elementary instructions performing computation and interacting with the storage keys. Some of these operations \emph{access} a given target storage key (e.g., read its value, write to it).
For simplicity, we assume that the following are known in advance and supplied with the transaction:

\begin{itemize}[topsep=0pt,itemsep=-1ex,partopsep=1ex,parsep=1ex]
    \item The non-empty set of storage keys $K \subseteq \keyset$ the transaction accesses, or an overestimate.\footnote{In Ethereum, transaction accesses are generally not known in advance. Instead, transactions can execute arbitrary logic (constrained by a maximum amount of computational effort) and their execution depends on the blockchain's state at the time of execution. Ethereum currently supports optional \textit{access lists} that allow transactions to specify their accesses~\cite{eip2930}. This optional list could be made mandatory to provide an overestimate of accesses. Additionally, it is worth noting that in other blockchains (e.g., Solana~\cite{solana_transactions} and Sui~\cite{sui_object_model}), an overestimate of accesses is typically known in advance.}
    \item The total time $t > 0$ it takes to execute the transaction on a single thread.\footnote{In reality, the time to execute a transaction depends on the hardware it is executed on. Thus, time is estimated in Ethereum by assigning each operation a computational effort. Then the sum of the computational effort of a transaction's operation can be seen as a proxy for time.}
\end{itemize}

Transactions execute concurrently, but \emph{atomically}, meaning that the overall effect of executing a batch of transactions should also be achievable by a sequential, single-core execution. For simplicity, we restrict to \emph{concurrent schedules} following a \emph{simple lock-based execution policy}: each storage key has a lock associated with it; whenever a thread wants to execute a transaction, it first locks all required storage keys, then executes the transaction, and then releases the locks. We assume that acquiring (and similarly releasing) the required locks happens simultaneously and takes no additional time.
These simplifications have a desirable side-effect: $t$ and $K$ for each transaction now uniquely determine the set of admissible concurrent schedules, allowing us to ignore other details about the transactions:  

\begin{definition}[Transaction]\label{def:tx}
    A transaction $\tx$ is specified through a tuple $(\txtime, \txkeys)$, where $\txtime > 0$ denotes the time required to execute the transaction and $\txkeys \neq \varnothing$ represents the set of storage keys the transaction demands, which are locked for the duration of the transaction. We will sometimes write $t(\tx)$ and $K(\tx)$ for $t$ and $K$ respectively.
\end{definition}\vspace{2pt}

To give the tuple associated with a transaction, we will write $\tx \simeq (t, R)$. Note that different transactions might have the same associated tuple. We will also write $\tx_1 \simeq \tx_2$ to mean that the two transactions have the same associated tuple.

\noindent
\begin{minipage}[b]{0.65\textwidth}
    \quad In Figure~\ref{fig:tx}, we illustrate a sample transaction $\tx \simeq (3,\{k_2,k_3\})$. Transaction $\tx$ thus takes 3 units of time to execute and utilizes storage keys $k_2$ and $k_3$. We illustrate this by a rectangle of corresponding length (i.e., time) and width (i.e., storage keys). Throughout, we will illustrate transactions in this manner to aid in visualizing concepts and results. Note that we will always represent transactions as rectangles (i.e., they use consecutive storage keys). In reality, this is, of course, not the case. Importantly, all our results also hold in the more general setting where a transaction can use any subset of storage keys. 
\end{minipage}%
\hfill
\begin{minipage}[b]{0.3\textwidth}
    \centering
    \begin{tikzpicture}[scale=1]
    \draw[-stealth, thick] (0, 0) -- (0, -2.3)  node[midway,above, rotate=180,sloped,inner sep = 20 pt] {\small time};
    \draw[ thick] (-0.1, -1.5) -- (0, -1.5)  node[above, rotate=180,sloped,inner sep = 20 pt] {};
    \node[text=gray] at (-0.3,-1.5)[](R1) {\small 3} ;
    \node[text=gray] at (0.25,0.25)[](R1) {\small $k_1$} ;
    \node[text=gray] at (0.75,0.25)[](R2) {\small $k_2$} ;
    \node[text=gray] at (1.25,0.25)[](R3) {\small $k_3$} ;
    \node[text=gray] at (1.75,0.25)[](R4) {\small $k_4$} ;
    \draw[-stealth, thick] (0, 0) -- (2.3,0)  node[midway,above,sloped,inner sep = 15 pt] {\small storage keys};

    \draw [dotted] (0.5,0) -- (0.5,-2.2);
    \draw [dotted] (1,0) -- (1,-2.2);
    \draw [dotted] (1.5,0) -- (1.5,-2.2);
    \draw [dotted] (2,0) -- (2,-2.2);

    \fill[color1!50] (0.5,-0) rectangle (1.5,-1.5); 
    \draw[black] (0.5,-0) rectangle (1.5,-1.5);
    \draw[white] (4,-0) rectangle (4,-0);
\end{tikzpicture}\vspace{-40pt}
    \captionof{figure}{Illustration of a sample transaction $\tx \simeq (3,\{k_2,k_3\})$.}
    \label{fig:tx}\vspace{5pt}
\end{minipage}

A \emph{blockchain} is a sequence of \emph{blocks} comprising bundles of transactions: $B_1, \dots, B_m$. The system starts in some predetermined initial state. The blocks are executed in order, starting from the oldest (the \emph{genesis} block $B_1$), successively changing the system's state. For simplicity, we assume no cross-block parallelism, so the execution of a block only starts after the previous block's execution has been completed. However, the execution of transactions inside a block happens concurrently. The system's state after executing $B_m$ is the \emph{current} system state. Users wanting to change the system's state compete for inclusion in the next block $B_{m + 1}$ and have to pay a \emph{fee} if successfully included. Desirably, the fee should be higher for more complex transactions and higher during periods of high demand due to limited block space. These requirements are typically decoupled and ensured through different means:

\begin{enumerate}[topsep=0pt,itemsep=-1ex,partopsep=1ex,parsep=1ex]
  \item \label{item:1} A transaction's complexity is quantified in units of \emph{gas}:\footnote{For Ethereum, the unit is called wei.} the more complex a transaction is, the more gas it consumes. Gas encompasses multiple components such as execution, storage space, and data bandwidth. For our purposes, we will only be concerned with \emph{execution gas}.\footnote{From this point forward, gas will refer specifically to execution gas.} Currently, the execution gas acts as a proxy for the execution time $t$,\footnote{Throughout we will assume this approximation to the exact.} but this need not be the case, as we will demonstrate in our paper.
  \item \label{item:2} The fair competition for block space is ensured through a \emph{transaction fee mechanism (\TFM)}: users submit transactions they would like to be included in the next block together with bids of what they would be willing to pay per unit of gas. 
  Importantly, block space is limited, i.e., there is limited space for transactions. The mechanism then determines the set\footnote{In practice, blocks are ordered, and transactions often compete for earlier positions (particularly in MEV settings). Following a common simplification in the TFM literature, we model blocks as \emph{sets} of transactions. Accounting for intra-block ordering would require extending our model and constraining it to schedulers that honor the specified ordering. Our framework appears reasonably extensible to such settings, and we leave a detailed treatment to future work.} of transactions to be included in the block, together with a price
  per unit of gas to be paid by each included transaction (potentially not the same for all transactions).\footnote{There are several other components of a \TFM, but for the level of detail we need here, this suffices.} 
\end{enumerate}

As such, an included transaction consuming $g$ gas units at a price of $p$ per unit of gas will have to pay a fee of $g \cdot p$ (generally in the blockchain's native currency).

To keep the separation of concerns in \eqref{item:1} and \eqref{item:2}, we will keep the formula for the fee $g \cdot p$ and instead vary the \emph{gas computation mechanism} used to determine $g$:

\begin{definition}[Gas Computation Mechanism]\label{def:gas-comp-mech}
    A gas computation mechanism (\GCM) takes as input a set of transactions $\txset$ to be included in a block and determines in a deterministic\footnote{We focus solely on deterministic mechanisms. This is common in blockchain fee markets, given the inherent difficulty of accessing true randomness on-chain (though notable approaches exist, such as VRFs). While randomized mechanisms can statistically improve over their deterministic counterparts, their ex-post behavior is often hard to predict, leading to higher fee uncertainty, an outcome generally undesirable in blockchain settings.} manner the amount of gas consumed by each transaction $\tx \in \txset$, written $\gas_{\txset}(\tx)$.
\end{definition}\vspace{2pt}

Currently deployed GCMs associate a fixed, predetermined gas consumption with each transaction, independent of the specific storage keys accessed by the transaction, making them unsuitable for a parallel execution environment. In particular, this is true for Ethereum's \emph{current} GCM:

\begin{definition}[Current \GCM]\label{def:current-gas-comp-mech}
    Given a set of transactions $\txset$ and a transaction $\tx \in T$ with $\tx \simeq (t, R)$, the \emph{current} \GCM computes the amount of gas used by $\tx$ as follows:  %
    \[
        \gas_{\txset}^C(\tx) := t.
    \]
    Since this does not depend on $T$, we often drop the subscript.
\end{definition}\vspace{2pt}

Our goal is to ensure that fees accurately reflect the parallelizability of transactions. Therefore, the gas consumption calculated for a transaction should depend on the set of storage keys it accesses and may also be influenced by factors external to the transaction itself (like interactions between transactions).

\section{\GCM Properties} \label{sect:gcm-prop}

Next, we outline several desirable properties that a \GCM should possess to provide the right incentives for parallelization. These properties should be viewed as a wishlist --- as will become clear later on, no single mechanism can satisfy all of them simultaneously.

We begin with two natural monotonicity properties, one for storage keys and one for time. First, a transaction that requires a subset of the storage keys used by another transaction, while taking the same amount of time, should consume no more gas (\cref{prop:resource_monotonicity,fig:resource_mon}). Similarly, a transaction that requires no more execution time than another, assuming both involve the same set of storage keys, should consume no more gas (\cref{prop:time_monotonicity,fig:time_mon}). 

\begin{figure}[t]
\centering
  
  \begin{minipage}[t]{0.48\linewidth}
  
  \centering
   \begin{tikzpicture}[scale=1]

  \draw[-stealth, thick] (0, 0) -- (0, -2.3)  node[midway,above, rotate=180,sloped,inner sep = 5 pt] {\small time};
  \node[text=gray] at (0.25,0.25)[](R1) {\small $k_1$} ;
  \node[text=gray] at (0.75,0.25)[](R2) {\small $k_2$} ;
  \node[text=gray] at (1.25,0.25)[](R3) {\small $k_3$} ;
  \node[text=gray] at (1.75,0.25)[](R4) {\small $k_4$} ;
  \draw[-stealth, thick] (0, 0) -- (2.3,0)  node[midway,above,sloped,inner sep = 15 pt] {\small storage keys};

  \draw [dotted] (0.5,0) -- (0.5,-2.2);
    \draw [dotted] (1,0) -- (1,-2.2);
  \draw [dotted] (1.5,0) -- (1.5,-2.2);
  \draw [dotted] (2,0) -- (2,-2.2);

    \fill[color1!50] (0.5,-0) rectangle (1,-1.5); 
    \draw[black] (0.5,-0) rectangle (1,-1.5);

\end{tikzpicture}\begin{tikzpicture}[scale=1]

  \draw[-stealth, thick] (0, 0) -- (0, -2.3)  node[midway,above, rotate=180,sloped,inner sep = 5 pt] {\small time};
  \node[text=gray] at (0.25,0.25)[](R1) {\small $k_1$} ;
  \node[text=gray] at (0.75,0.25)[](R2) {\small $k_2$} ;
  \node[text=gray] at (1.25,0.25)[](R3) {\small $k_3$} ;
  \node[text=gray] at (1.75,0.25)[](R4) {\small $k_4$} ;
  \draw[-stealth, thick] (0, 0) -- (2.3,0)  node[midway,above,sloped,inner sep = 15 pt] {\small storage keys};

  \draw [dotted] (0.5,0) -- (0.5,-2.2);
    \draw [dotted] (1,0) -- (1,-2.2);
  \draw [dotted] (1.5,0) -- (1.5,-2.2);
  \draw [dotted] (2,0) -- (2,-2.2);

    \fill[color2!50] (0.5,-0) rectangle (1.5,-1.5); 
    \draw[black] (0.5,-0) rectangle (1.5,-1.5);
 
\end{tikzpicture}\vspace{-4pt}
    \caption{Illustration of \cref{prop:resource_monotonicity},  where $\txset=\emptyset$, \fcolorbox{black}{color1!50}{\rule{0pt}{3pt}\rule{3pt}{0pt}}\hspace{1pt} represents a sample $\tx_1$ and \fcolorbox{black}{color2!50}{\rule{0pt}{3pt}\rule{3pt}{0pt}}\hspace{1pt} represents a sample $\tx_2$. } \label{fig:resource_mon}
  \end{minipage}%
  \hfill
  \begin{minipage}[t]{0.48\linewidth}
  
  \centering
   \begin{tikzpicture}[scale=1]

  \draw[-stealth, thick] (0, 0) -- (0, -2.3)  node[midway,above, rotate=180,sloped,inner sep = 5 pt] {\small time};
  \node[text=gray] at (0.25,0.25)[](R1) {\small $k_1$} ;
  \node[text=gray] at (0.75,0.25)[](R2) {\small $k_2$} ;
  \node[text=gray] at (1.25,0.25)[](R3) {\small $k_3$} ;
  \node[text=gray] at (1.75,0.25)[](R4) {\small $k_4$} ;
  \draw[-stealth, thick] (0, 0) -- (2.3,0)  node[midway,above,sloped,inner sep = 15 pt] {\small storage keys};

  \draw [dotted] (0.5,0) -- (0.5,-2.2);
    \draw [dotted] (1,0) -- (1,-2.2);
  \draw [dotted] (1.5,0) -- (1.5,-2.2);
  \draw [dotted] (2,0) -- (2,-2.2);

    \fill[color1!50] (0.5,-0) rectangle (1.5,-1.5); 
    \draw[black] (0.5,-0) rectangle (1.5,-1.5);

\end{tikzpicture}\begin{tikzpicture}[scale=1]

  \draw[-stealth, thick] (0, 0) -- (0, -2.3)  node[midway,above, rotate=180,sloped,inner sep = 5 pt] {\small time};
  \node[text=gray] at (0.25,0.25)[](R1) {\small $k_1$} ;
  \node[text=gray] at (0.75,0.25)[](R2) {\small $k_2$} ;
  \node[text=gray] at (1.25,0.25)[](R3) {\small $k_3$} ;
  \node[text=gray] at (1.75,0.25)[](R4) {\small $k_4$} ;
  \draw[-stealth, thick] (0, 0) -- (2.3,0)  node[midway,above,sloped,inner sep = 15 pt] {\small storage keys};

  \draw [dotted] (0.5,0) -- (0.5,-2.2);
    \draw [dotted] (1,0) -- (1,-2.2);
  \draw [dotted] (1.5,0) -- (1.5,-2.2);
  \draw [dotted] (2,0) -- (2,-2.2);

    \fill[color2!50] (0.5,-0) rectangle (1.5,-2.2); 
    \draw[black] (0.5,-0) rectangle (1.5,-2.2);

\end{tikzpicture}\vspace{-4pt}
    \caption{Illustration of \cref{prop:time_monotonicity},  where $\txset=\emptyset$, \fcolorbox{black}{color1!50}{\rule{0pt}{3pt}\rule{3pt}{0pt}}\hspace{1pt} represents a sample $\tx_1$ and \fcolorbox{black}{color2!50}{\rule{0pt}{3pt}\rule{3pt}{0pt}}\hspace{1pt} represents a sample $\tx_2$. } \label{fig:time_mon}
  \end{minipage}%

    \begin{minipage}[t]{0.48\linewidth}
  \centering
   \begin{tikzpicture}[scale=1]

  \draw[-stealth, thick] (0, 0) -- (0, -2.3)  node[midway,above, rotate=180,sloped,inner sep = 5 pt] {\small time};
  \node[text=gray] at (0.25,0.25)[](R1) {\small $k_1$} ;
  \node[text=gray] at (0.75,0.25)[](R2) {\small $k_2$} ;
  \node[text=gray] at (1.25,0.25)[](R3) {\small $k_3$} ;
  \node[text=gray] at (1.75,0.25)[](R4) {\small $k_4$} ;
  \draw[-stealth, thick] (0, 0) -- (2.3,0)  node[midway,above,sloped,inner sep = 15 pt] {\small storage keys};

  \draw [dotted] (0.5,0) -- (0.5,-2.2);
    \draw [dotted] (1,0) -- (1,-2.2);
  \draw [dotted] (1.5,0) -- (1.5,-2.2);
  \draw [dotted] (2,0) -- (2,-2.2);

    \fill[color1!50] (0.5,-0) rectangle (1,-1.5); 
    \draw[black] (0.5,-0) rectangle (1,-1.5);

    \fill[color2!50] (1,-0) rectangle (2,-1); 
    \draw[black] (1,-0) rectangle (2,-1);

\end{tikzpicture}\begin{tikzpicture}[scale=1]

  \draw[-stealth, thick] (0, 0) -- (0, -2.3)  node[midway,above, rotate=180,sloped,inner sep = 5 pt] {\small time};
  \node[text=gray] at (0.25,0.25)[](R1) {\small $k_1$} ;
  \node[text=gray] at (0.75,0.25)[](R2) {\small $k_2$} ;
  \node[text=gray] at (1.25,0.25)[](R3) {\small $k_3$} ;
  \node[text=gray] at (1.75,0.25)[](R4) {\small $k_4$} ;
  \draw[-stealth, thick] (0, 0) -- (2.3,0)  node[midway,above,sloped,inner sep = 15 pt] {\small storage keys};

  \draw [dotted] (0.5,0) -- (0.5,-2.2);
    \draw [dotted] (1,0) -- (1,-2.2);
  \draw [dotted] (1.5,0) -- (1.5,-2.2);
  \draw [dotted] (2,0) -- (2,-2.2);

    \fill[color1!50] (0.5,-0) rectangle (1,-1.5); 
    \draw[black] (0.5,-0) rectangle (1,-1.5);

    \fill[color2!50] (1,-0) rectangle (2,-1); 
    \draw[black] (1,-0) rectangle (2,-1); 

    \fill[color3!50] (1,-1) rectangle (2,-1.8); 
    \draw[black] (1,-1) rectangle (2,-1.8);

\end{tikzpicture}\vspace{-4pt}
    \caption{Illustration of \cref{prop:set_inclusion}, where $\txset=\emptyset$,  $\txset_1=\{\tx_1, tx_2\}$, and $\txset_2=\txset_1\cup \{\tx_3\}$. Here, \fcolorbox{black}{color1!50}{\rule{0pt}{3pt}\rule{3pt}{0pt}}\hspace{1pt} represents a sample $\tx_1$, \fcolorbox{black}{color2!50}{\rule{0pt}{3pt}\rule{3pt}{0pt}}\hspace{1pt} represents a sample $\tx_2$, and \fcolorbox{black}{color3!50}{\rule{0pt}{3pt}\rule{3pt}{0pt}}\hspace{1pt} represents a sample $\tx_3$.} \label{fig:set}
  \end{minipage}%
  \hfill
    \begin{minipage}[t]{0.48\linewidth}
  
  \centering
   \begin{tikzpicture}[scale=1]

  \draw[-stealth, thick] (0, 0) -- (0, -2.3)  node[midway,above, rotate=180,sloped,inner sep = 5 pt] {\small time};
  \node[text=gray] at (0.25,0.25)[](R1) {\small $k_1$} ;
  \node[text=gray] at (0.75,0.25)[](R2) {\small $k_2$} ;
  \node[text=gray] at (1.25,0.25)[](R3) {\small $k_3$} ;
  \node[text=gray] at (1.75,0.25)[](R4) {\small $k_4$} ;
  \draw[-stealth, thick] (0, 0) -- (2.3,0)  node[midway,above,sloped,inner sep = 15 pt] {\small storage keys};

  \draw [dotted] (0.5,0) -- (0.5,-2.2);
    \draw [dotted] (1,0) -- (1,-2.2);
  \draw [dotted] (1.5,0) -- (1.5,-2.2);
  \draw [dotted] (2,0) -- (2,-2.2);

    \fill[color1!20] (0.5,-0) rectangle (2,-2); 
    \fill[color1!50] (0.5,-0) rectangle (1,-0.5); 
    \fill[color1!50] (1,-0.5) rectangle (1.5,-1.5); 
    \fill[color1!50] (1.5,-1.5) rectangle (2,-2); 
    \draw[black] (0.5,-0) rectangle (2,-2);

\end{tikzpicture}\begin{tikzpicture}[scale=1]

  \draw[-stealth, thick] (0, 0) -- (0, -2.3)  node[midway,above, rotate=180,sloped,inner sep = 5 pt] {\small time};
  \node[text=gray] at (0.25,0.25)[](R1) {\small $k_1$} ;
  \node[text=gray] at (0.75,0.25)[](R2) {\small $k_2$} ;
  \node[text=gray] at (1.25,0.25)[](R3) {\small $k_3$} ;
  \node[text=gray] at (1.75,0.25)[](R4) {\small $k_4$} ;
  \draw[-stealth, thick] (0, 0) -- (2.3,0)  node[midway,above,sloped,inner sep = 15 pt] {\small storage keys};

  \draw [dotted] (0.5,0) -- (0.5,-2.2);
    \draw [dotted] (1,0) -- (1,-2.2);
  \draw [dotted] (1.5,0) -- (1.5,-2.2);
  \draw [dotted] (2,0) -- (2,-2.2);

    \fill[color2!20] (0.5,-0) rectangle (1.5,-1); 
    \fill[color2!50] (0.5,-0) rectangle (1,-0.5); 
    \fill[color2!50] (1,-0.5) rectangle (1.5,-1); 
    \draw[black]   (0.5,-0) rectangle (1.5,-1); 

    \fill[color3!20] (1,-1) rectangle (2,-2); 
    
    \fill[color3!50] (1,-1) rectangle (1.5,-1.5); 
    \fill[color3!50] (1.5,-1.5) rectangle (2,-2); 
    \draw[black]    (1,-1) rectangle (2,-2);
\end{tikzpicture}\vspace{-4pt}
    \caption{Illustration of \cref{prop:tx_bundling}, where $\txset=\emptyset$, \fcolorbox{black}{color1!20}{\rule{0pt}{3pt}\rule{3pt}{0pt}}\hspace{1pt} represents a sample $\tx_3$, while \fcolorbox{black}{color2!20}{\rule{0pt}{3pt}\rule{3pt}{0pt}}\hspace{1pt} represents a sample $\tx_1$ and  \fcolorbox{black}{color3!20}{\rule{0pt}{3pt}\rule{3pt}{0pt}}\hspace{1pt} represents a sample $\tx_2$. The darker shaded areas indicate when a transaction operates on a storage key.} \label{fig:decom}
  \end{minipage}%

\end{figure}

\begin{property}[Storage Key Monotonicity]\label{prop:resource_monotonicity}
    Given a set of transactions $\txset$ and two transactions $\tx_1 \simeq (\txtime_1, \txkeys_1)$ and $\tx_2 \simeq (\txtime_2, \txkeys_2)$, both not in $T$,
    such that $\txtime_1=\txtime_2$ and $\txkeys_1 \subseteq \txkeys_2$:
    $$\gas_{T \cup \{\tx_1\}}(\tx_1) \leq \gas_{T \cup \{\tx_2\}}(\tx_2). $$
\end{property}

\begin{property}[Time Monotonicity]\label{prop:time_monotonicity}
   Given a set of transactions $\txset$ and two transactions $\tx_1 \simeq (\txtime_1, \txkeys_1)$ and $\tx_2 \simeq (\txtime_2, \txkeys_2)$, both not in $T$,
   such that $\txtime_1 \leq \txtime_2$ and $\txkeys_1 = \txkeys_2$:
    $$\gas_{T \cup \{\tx_1\}}(\tx_1) \leq \gas_{T \cup \{\tx_2\}}(\tx_2). $$
\end{property}

The previous two properties fix one dimension while varying the other. One can also define a seemingly stronger property that allows both to vary, as follows:

\begin{property}[Storage Key-Time Monotonicity]\label{prop:resource_time_monotonicity}
    Given a set of transactions $\txset$ and two transactions $\tx_1 \simeq (\txtime_1, \txkeys_1)$ and $\tx_2 \simeq (\txtime_2, \txkeys_2)$, both not in $T$,
    such that $\txtime_1 \leq \txtime_2$ and $\txkeys_1 \subseteq \txkeys_2$:
    $$\gas_{T \cup \{\tx_1\}}(\tx_1) \leq \gas_{T \cup \{\tx_2\}}(\tx_2). $$
\end{property}

Intuitively, if $\tx_1 \lesssim \tx_2$, by which we mean $t(\tx_1) \leq t(\tx_2)$ and $K(\tx_1) \subseteq K(\tx_2)$, then $\tx_1$ should cost no less than $\tx_2$. However, an attentive reader will observe that the former two are collectively equivalent to the latter (the proof and all subsequent omitted proofs can be found in the appendix):

\begin{restatable}{lemma}{lemmaproponetwothree} \label{lemma:prop-1-2-3-equivalence} \cref{prop:resource_time_monotonicity} holds
if and only if 
\cref{prop:resource_monotonicity,prop:time_monotonicity} hold.
\end{restatable}

Let us now take a moment to briefly evaluate why \cref{prop:resource_monotonicity,prop:time_monotonicity,prop:resource_time_monotonicity} are not merely intuitive, but their violation can lead to harmful consequences and misaligned incentives: suppose $\tx_1 \lesssim \tx_2$ and $T$ is a set of transactions containing neither of the two. If $\gas_{T \cup \{\tx_1\}}(\tx_1) > \gas_{T \cup \{\tx_2\}}(\tx_2)$, a user intending to submit $\tx_1$ might instead pad $\tx_1$ with unnecessary instructions and declare a larger access list to decrease the gas consumption. This would paradoxically reduce the gas usage and, assuming a reasonable $\TFM$ is used to compute transaction fees, also lower the transaction fee.

For any of the three properties above, we say that the property is \emph{strictly} satisfied if for $\tx_1 \not\simeq \tx_2$ the conclusion inequality holds strictly.
Note that properties holding strictly are even more desirable with respect to the reasoning above: replacing a transaction with a ``larger'' one is then not only no better but actively worse.
Unsurprisingly, \cref{lemma:prop-1-2-3-equivalence} also holds for the strict versions of the properties:

\begin{restatable}{lemma}{lemmaproponetwothreeequiv}\label{lemma:prop-1-2-3-equivalence-strict} \cref{prop:resource_time_monotonicity} holds strictly
if and only if 
\cref{prop:resource_monotonicity,prop:time_monotonicity} hold strictly.
\end{restatable}

Next, we introduce another desirable monotonicity property, this time with a different emphasis: if two sets of transactions satisfy $\txset_1 \subseteq \txset_2$, the transactions in $\txset_1$ should collectively consume no more gas than those in $\txset_2$
(\cref{prop:set_inclusion,fig:set}). To formalize this, given a set of transactions $\txset$ and a subset $\txset' \subseteq \txset$, write $\gas_{\txset}(\txset') := \sum_{\tx' \in \txset'} \gas_{\txset}(\tx')$ for the total gas consumed by the transactions in $\txset'$ when included in the set $\txset$ that constitutes a block.

\begin{property}[Set Inclusion]\label{prop:set_inclusion}
    Given a set of transactions  $\txset$ and two sets of transactions $\txset_1 \subseteq \txset_2$, disjoint from $\txset$:
    \[
        \gas_{\txset \cup \txset_1}(\txset_1) \leq \gas_{\txset \cup \txset_2}(\txset_2).
    \]
\end{property}

To understand why this property is desirable, consider a \GCM for which it does not hold. Then, there must be a scenario where a set of transactions can reduce their total gas consumption by adding additional transactions. This situation could be exploited through collusion by the users originating these transactions. Importantly, such a possibility is undesirable, as it would primarily benefit sophisticated users capable of orchestrating such arrangements.

Similarly to before, we say that \cref{prop:set_inclusion} holds \emph{strictly} if for $\txset_1 \neq \txset_2$ the inequality in the conclusion holds strictly, which is desirable for reasons similar to those discussed above.

We now move on to more complex properties. Given two transactions $\tx_1 \simeq (\txtime_1, \txkeys_1)$ and $\tx_2 \simeq (\txtime_2, \txkeys_2),$ a transaction $\tx_3$ is their \emph{sequential composition} (or, more simply, \emph{concatenation}), 
if it executes the steps of $\tx_1$ followed by the steps of $\tx_2$. Note that, in this case $\tx_3 \simeq (t_1 + t_2, K_1 \cup K_2)$.
Our next property states that the concatenation of two transactions should consume no less gas than submitting them individually (\cref{prop:tx_bundling,fig:decom}). The two individual transactions perform the same actions as their concatenation, but not atomically, with no guarantee of their relative ordering or control over what happens between them. Naturally, enforcing atomicity and ordering limits the set of admissible concurrent schedules and requires storage keys to remain locked over longer continuous timespans. In particular, $\tx_3$
requires the storage keys in $\txkeys_1 \cup \txkeys_2$ to be locked over a continuous span of $\txtime_1 + \txtime_2$ units, while $\tx_1$ and $\tx_2$ submitted individually only require exclusive access to storage keys in $\txkeys_1$ for $t_1$ units and to storage keys in $\txkeys_2$ for $t_2$ units. Hence, the ``bigger'' transaction is at least as hard to schedule as its two constituent ``parts'' and should hence consume no less gas.

\begin{property}[Transaction Bundling]\label{prop:tx_bundling}
    Consider a set of transactions $\txset$ and three transactions $\tx_1, \tx_2, \tx_3 \notin T$ such that $\tx_3$ is the concatenation of $\tx_1$ and $\tx_2$, then:
    \[\gas_{\txset \cup \{\tx_1, \tx_2\}}(\tx_1)+\gas_{\txset \cup \{\tx_1, \tx_2\}}(\tx_2) \leq \gas_{\txset \cup \{\tx_3\}}(\tx_3).\]
 \end{property}

Let us again consider the risks of having a \GCM that does not satisfy the previous property. In such a scenario, a group of users could collude to collectively consume less gas by combining their transactions into a single transaction rather than processing them individually. This outcome would be undesirable, particularly because it disproportionately benefits sophisticated users, as we outlined before. Note that this could even be the case for a single user wanting to submit multiple transactions.

We say that \cref{prop:tx_bundling} holds \emph{strictly} if the inequality in the conclusion holds strictly, which is again more desirable than the basic version of the property.

Next, we would like to formalize the intuitive idea that a transaction's gas consumption fairly reflects its impact on the execution time.
To do so, we first need to formalize scheduling more precisely. Let $n \geq 2$ be the number of available threads. 

\begin{definition}
    A \emph{scheduler} (for $n$ threads) takes as input a set of transactions $T$, and outputs a \emph{concurrent schedule} using at most $n$ threads to execute all transactions in $T$. This schedule specifies the operations each thread should perform and the order in which they should be performed.
\end{definition}

The following conditions should hold for any generated schedule:
\begin{itemize}[topsep=0pt,itemsep=-1ex,partopsep=1ex,parsep=1ex]
    \item Each transaction is assigned to a single thread;
    \item A thread can work on only one transaction at a time;
    \item There is no preemption: once a thread starts executing a transaction, it completes the transaction without context switching;
    \item Transactions with overlapping storage key access sets cannot be executed in parallel: one must finish before the other can begin.
\end{itemize}

\begin{figure}[t]
\centering

\begin{tikzpicture}[scale=1]

  \draw[-stealth, thick] (0, 0) -- (0, -4.5)  node[midway,above, rotate=180,sloped,inner sep = 5 pt] {\small time};

  \node[text=gray] at (0.25,0.25)[](R1) {\small $k_1$} ;
  \node[text=gray] at (0.75,0.25)[](R2) {\small $k_2$} ;
  \node[text=gray] at (1.25,0.25)[](R3) {\small $k_3$} ;
  \node[text=gray] at (1.75,0.25)[](R4) {\small $k_4$} ;
  \node[text=gray] at (2.25,0.25)[](R5) {\small $k_5$} ;
  \node[text=gray] at (2.75,0.25)[](R6) {\small $k_6$} ;
  \node[text=gray] at (3.25,0.25)[](R7) {\small $k_7$} ;
  \node[text=gray] at (3.75,0.25)[](R8) {\small $k_8$} ;
  \draw[-stealth, thick] (0, 0) -- (4.5,0)  node[midway,above,sloped,inner sep = 15 pt] {\small storage keys};

    \draw [dotted] (0.5,0) -- (0.5,-4.2);
    \draw [dotted] (1,0) -- (1,-4.2);
    \draw [dotted] (1.5,0) -- (1.5,-4.2);
    \draw [dotted] (2,0) -- (2,-4.2);
    \draw [dotted] (2.5,0) -- (2.5,-4.2);
    \draw [dotted] (3,0) -- (3,-4.2);
    \draw [dotted] (3.5,0) -- (3.5,-4.2);
    \draw [dotted] (4,0) -- (4,-4.2);

    \fill[color1!50] (0.5,-0) rectangle (4,-1);
    \draw[black]  (0.5,-0) rectangle (4,-1);
    
    \fill[color2!50] (0.5,-1) rectangle (1.5,-3);
    \draw[black]  (0.5,-1) rectangle (1.5,-3);
    
    \fill[color3!50] (1.5,-1) rectangle (3,-3.5);
    \draw[black]  (1.5,-1) rectangle (3,-3.5);
    
    \fill[color4!50] (3,-1) rectangle (4,-2);
    \draw[black]  (3,-1) rectangle (4,-2);

\end{tikzpicture}\hspace{1cm}\begin{tikzpicture}[scale=1]

  \draw[-stealth, thick] (0, 0) -- (0, -4.5)  node[midway,above, rotate=180,sloped,inner sep = 5 pt] {\small time};

  \node[text=gray] at (0.25,0.25)[](R1) {\small $k_1$} ;
  \node[text=gray] at (0.75,0.25)[](R2) {\small $k_2$} ;
  \node[text=gray] at (1.25,0.25)[](R3) {\small $k_3$} ;
  \node[text=gray] at (1.75,0.25)[](R4) {\small $k_4$} ;
  \node[text=gray] at (2.25,0.25)[](R5) {\small $k_5$} ;
  \node[text=gray] at (2.75,0.25)[](R6) {\small $k_6$} ;
  \node[text=gray] at (3.25,0.25)[](R7) {\small $k_7$} ;
  \node[text=gray] at (3.75,0.25)[](R8) {\small $k_8$} ;
  \draw[-stealth, thick] (0, 0) -- (4.5,0)  node[midway,above,sloped,inner sep = 15 pt] {\small storage keys};

    \draw [dotted] (0.5,0) -- (0.5,-4.2);
    \draw [dotted] (1,0) -- (1,-4.2);
    \draw [dotted] (1.5,0) -- (1.5,-4.2);
    \draw [dotted] (2,0) -- (2,-4.2);
    \draw [dotted] (2.5,0) -- (2.5,-4.2);
    \draw [dotted] (3,0) -- (3,-4.2);
    \draw [dotted] (3.5,0) -- (3.5,-4.2);
    \draw [dotted] (4,0) -- (4,-4.2);

    \fill[color1!50] (0.5,-0) rectangle (4,-1);
    \draw[black]  (0.5,-0) rectangle (4,-1);
    
    \fill[color2!50] (0.5,-1) rectangle (1.5,-3);
    \draw[black]  (0.5,-1) rectangle (1.5,-3);
    
    \fill[color3!50] (1.5,-1) rectangle (3,-3.5);
    \draw[black]  (1.5,-1) rectangle (3,-3.5);
    
    \fill[color4!50] (3,-3) rectangle (4,-4);
    \draw[black]  (3,-3) rectangle (4,-4);

\end{tikzpicture}\vspace{-4pt}
\caption{Illustration of optimal schedules for a set of four transactions: $\tx_1 \simeq (2,\{k_2,\dots,k_8\})$  (shown in \fcolorbox{black}{color1!50}{\rule{0pt}{3pt}\rule{3pt}{0pt}}\hspace{1pt}), $\tx_2 \simeq (4,\{k_2,k_3\})$  (shown in \fcolorbox{black}{color2!50}{\rule{0pt}{3pt}\rule{3pt}{0pt}}\hspace{1pt}), $\tx_3 \simeq (5,\{k_4,k_5,k_6\})$  (shown in \fcolorbox{black}{color3!50}{\rule{0pt}{3pt}\rule{3pt}{0pt}}\hspace{1pt}), and $\tx_4 \simeq (2,\{k_7,k_8\})$  (shown in \fcolorbox{black}{color4!50}{\rule{0pt}{3pt}\rule{3pt}{0pt}}\hspace{1pt}). In the left plot, we show the optimal schedule for $n=3$, and in the right plot for $n=2$. Notice how for $n=2$, we cannot schedule the three transactions $\tx_2$, $\tx_3$, and $\tx_4$ to be executed in parallel even though they access pairwise-disjoint sets of storage keys.} \label{fig:schedule}

\end{figure}

Note that some of these conditions are natural for scheduling in general, while others arise from us assuming the \emph{simple lock-based execution policy}.\footnote{Using locks is one way to enforce this policy, but in our case --- where the contents of a block are known before execution commences --- it can also be achieved without locks, as long as the execution environment ensures that threads strictly follow a pre-determined schedule. We chose this name because it is largely suggestive of the intended semantics.} For our purposes, we are not concerned with which thread executes which transaction, but only that no more than $n$ transactions ever execute simultaneously (so we can draw concurrent schedules as in \cref{fig:schedule}, which illustrates the concept). To determine gas consumption, our notation will need to capture even less --- given a scheduler $s$ (for $n$ threads), we write $v^s(T)$ for the \emph{makespan} of the schedule produced by $s$ for the set of transactions $T$ (i.e., the time required to execute the schedule in parallel using $n$ threads). It is instructive to read the paper having in mind as a prime example $s$ being the \emph{optimal scheduler} (for $n$ threads), which returns a schedule for $n$ threads
that minimizes the makespan. However, implementing such a scheduler is computationally prohibitive in practice, so greedy heuristics are typically used instead. To make our results apply even to non-optimal schedulers, we assume a set of minimal, reasonable properties for the scheduler:

\begin{enumerate}[label=(S\arabic*),topsep=0pt,itemsep=-1ex,partopsep=1ex,parsep=1ex,leftmargin=2.25em]
    \item \label{sched:monot-t} \emph{Monotonicity in $\txset$:} for any $\txset \subseteq \txset'$, we have $v^s(\txset) \leq v^s(\txset')$. Intuitively, scheduling no fewer transactions takes no less time.
    \item \label{sched:monot-bundle} \emph{Monotonicity under bundling:} consider any set of transactions $\txset$ and three transactions $\tx_1, \tx_2, \tx_3 \notin \txset$, such that $\tx_3$ is the concatenation of $\tx_1$ and $\tx_2$, then we have $v^s(\txset \cup \{\tx_1, \tx_2\}) \leq v^s(\txset \cup \{\tx_3\})$. Intuitively, replacing two transactions by their concatenation makes scheduling no easier. 
    \item \label{sched:monot-t-R} \emph{Monotonicity in $t$ and $K$:} consider any set of transactions $\txset$ and two transactions $\tx_1 \simeq (\txtime_1, \txkeys_1)$ and $\tx_2 \simeq (\txtime_2, \txkeys_2)$, both not in $T$,
    such that $\txtime_1 \leq \txtime_2$ and $\txkeys_1 \subseteq \txkeys_2$, i.e., $\tx_1 \lesssim \tx_2$, then we have  
    $v(T \cup \{\tx_1\}) \leq v(T \cup \{\tx_2\})$. Intuitively, ``larger'' transactions are no easier to schedule. 
    \item \label{sched:empty-set} \emph{Empty set:} $v(\varnothing) = 0$.
\end{enumerate}

The optimal scheduler can be easily seen to satisfy these properties.\footnote{All but the second straightforward, while for the second, any admissible schedule for $\txset \cup \{\tx_3\}$ can be turned into a same-makespan admissible schedule for $\txset \cup \{\tx_1, \tx_2\}$ by replacing $\tx_3$ by $\tx_1$ followed immediately by $\tx_2$.} \emph{Our positive results will apply to any scheduler satisfying the properties, while our negative results will be for the optimal scheduler itself.} Henceforth, we drop the $s$ superscript for brevity and follow this convention for who $s$ is.

An attentive reader may notice that \ref{sched:monot-t}--\ref{sched:monot-t-R} partially resemble \cref{prop:resource_monotonicity,prop:time_monotonicity,prop:resource_time_monotonicity,prop:set_inclusion,prop:tx_bundling} for \GCMs. This resemblance is not coincidental, but it is important to emphasize that they address different aspects: the former are properties of the scheduler, while the latter are properties of the \GCM. A \GCM can, in fact, satisfy \cref{prop:resource_monotonicity,prop:time_monotonicity,prop:resource_time_monotonicity,prop:set_inclusion,prop:tx_bundling} without depending on the scheduler at all (e.g., the current mechanism). Conversely, for mechanisms defined in terms of the scheduler, \ref{sched:monot-t}--\ref{sched:monot-t-R} play a crucial role in establishing their properties, including \cref{prop:resource_monotonicity,prop:time_monotonicity,prop:resource_time_monotonicity,prop:set_inclusion,prop:tx_bundling}.

Armed as such, we now return to our latest goal: formalizing the idea that a transaction's gas consumption should fairly reflect its impact on execution time. We do this as follows:

\begin{property}[Scheduling Monotonicity]\label{prop:schedule_monotonicity}
    Given a set of transactions $\txset$ and two transactions $\tx_1$ and $\tx_2$, both not in $T$,
    such that 
    $v(\txset\cup\{\tx_1\}) < v(\txset\cup\{\tx_2\})$:\footnote{Interestingly, unlike our earlier properties, writing $v(\txset\cup\{\tx_1\}) \leq v(\txset\cup\{\tx_2\})$ here may be too strong, as applying the property twice would then imply that if $v(\txset\cup\{\tx_1\}) = v(\txset\cup\{\tx_2\})$, then $\gas_{\txset \cup \{\tx_1\}}(\tx_1) = \gas_{\txset \cup \{\tx_2\}}(\tx_2)$, which is not necessarily desirable. Writing $v(\txset\cup\{\tx_1\}) \leq v(\txset\cup\{\tx_2\})$ would also unintentionally make the strict and non-strict versions of the property incomparable.} 
    \[ \gas_{\txset \cup \{\tx_1\}}(\tx_1) \leq \gas_{\txset \cup \{\tx_2\}}(\tx_2).\]
\end{property}

Intuitively, transactions with higher marginal contributions to the execution time should consume no less gas.\footnote{Technically, what we wrote above in \cref{prop:schedule_monotonicity} are not marginal contributions, but can be made to be by subtracting $v(\txset)$ from both sides of the inequality in the antecedent.} As usual, we define a \emph{strict} version of this property, where the latter inequality also becomes strict (i.e., higher marginal contributions in execution time imply higher gas consumption). One might be tempted to believe that Scheduling Monotonicity implies Storage Key-Time Monotonicity. However, proving this requires a non-strict inequality in the premise $v(\txset\cup\{\tx_1\}) < v(\txset\cup\{\tx_2\})$.

There is a second way in which gas consumption should adequately indicate the effort required to execute transactions (\cref{prop:efficiency}).  Specifically, the gas consumption of all transactions in a block should \emph{collectively} account for the total time needed to execute the block. This can be seen as properly tracking the time consumed by the execution environment to execute the block.

\begin{property}[Efficiency]\label{prop:efficiency} Consider a set of transactions $\txset$ and recall the definition $\gas_{\txset}(\txset) = \sum_{\tx \in \txset}\gas_{\txset}(\tx)$, then:
\[
    \gas_{\txset}(\txset) = v(\txset).
\] 
\end{property}

Next, we introduce two practical properties. The first requires that transaction submitters be able to estimate a transaction's gas consumption in advance. We formalize this by requiring that $\gas_\txset(\tx)$ does not depend on $\txset$ (\cref{prop:Composability}).

\begin{property}[Easy Gas Estimation]\label{prop:Composability} Given two sets of transactions $\txset_1$ and $\txset_2$ and a transaction $\tx$ belonging to neither $T_1$ nor $T_2$:
\[
    \gas_{\txset_1 \cup \{\tx\}}(\tx) = 
    \gas_{\txset_2 \cup \{\tx\}}(\tx).
\]
\end{property}

This property ensures a good user experience by making it easy for users to estimate gas usage. If the gas consumption of a transaction depended on the remaining transactions in the block, this estimation would become significantly more complex. In general, we aim to keep gas estimation straightforward to avoid giving an advantage to more sophisticated users. Additionally, as we will see later, a \GCM satisfying this property can be seamlessly composed with existing \TFMs, retaining their desirable properties (see \cref{sec:outlook}). Sadly, as one might have already guessed, Easy Gas Estimation is incompatible with Scheduling Monotonicity (except for \emph{constant} mechanisms, i.e., \GCMs that return a constant independent of $T$ and $\tx$), and also incompatible with Efficiency:

\begin{theorem}\label{thm:impossibility} Easy Gas Estimation (\cref{prop:Composability}) is:
\begin{enumerate}
    \item Incompatible with Scheduling Monotonicity (\cref{prop:schedule_monotonicity}) unless using a constant \GCM.\footnote{Constant \GCMs satisfy Scheduling Monotonicity, but not Strict Scheduling Monotonicity.}
    \item Incompatible with Efficiency (\cref{prop:efficiency}). 
\end{enumerate}
\end{theorem}
\begin{proof} 

Consider a mechanism $M$ with Easy Gas Estimation; i.e., $\gas^M(\tx)$ is meaningful without a subscript for the set of transactions $T$ in the block. Then:

\begin{enumerate}[topsep=0pt,itemsep=-1ex,partopsep=1ex,parsep=1ex]
    \item Assume that $M$ has Scheduling Monotonicity; we will show that $M$ is a constant mechanism.

    Call two transactions $\tx_1 \simeq (t_1, K_1)$ and $\tx_2 \simeq (t_2, K_2)$ \emph{incomparable} if neither $K_1 \subseteq K_2$ nor $K_2 \subseteq K_1$. As a first step, we will show that if $\tx_1$ and $\tx_2$ are incomparable, then $\gas^M(\tx_1) = \gas^M(\tx_2)$. To show this, assume that $K_1 \not\subseteq K_2$. We will show that then $\gas^M(\tx_1) \geq \gas^M(\tx_2)$. Exchanging the roles of $\tx_1$ and $\tx_2$ will then give the conclusion. Let $k \in K_1 \setminus K_2$ and $\tx_3$ be another transaction such that $\tx_3 \simeq (t_3, \{k\})$ for some $t_3 > t_2$, and define $T = \{\tx_3\}$. Then, for any number of threads $n \geq 2$ and any scheduler that is optimal for two transactions, we have $v(T \cup \{\tx_1\}) = t_1 + t_3 > t_3 = v(T \cup \{\tx_2\})$. By Scheduling Monotonicity, this implies that $\gas^M(\tx_1) \geq \gas^M(\tx_2)$.
    
    We now know that $M$ associates the same gas consumption to any pair of incomparable transactions. Armed as such, consider two arbitrary transactions $\tx_1 \simeq (t_1, K_1)$ and $\tx_2 \simeq (t_2, K_2)$. Let $k$ be an arbitrary storage key \emph{not} in $K_1 \cup K_2$ and $\tx_3$ be a transaction such that $\tx_3 \simeq (1, \{r\})$. Since $K_1$ and $K_2$ are non-empty, one can easily see that $\tx_3$ is incomparable with both $\tx_1$ and $\tx_2$, from which $\gas^M(\tx_1) = \gas^M(\tx_3) = \gas^M(\tx_2)$.

    \item Assume for a contradiction that $M$ has Efficiency, then for any transaction $\tx \simeq (\txtime, \txkeys)$ we have $\gas^M(\tx) = \gas_\varnothing^M(\{\tx\}) = v(\{\tx\}) = \txtime$. Hence, by definition, $M$ is the current mechanism, which can be easily seen to not satisfy Efficiency: assume $n \geq 2$ threads and consider the set of transactions $T = \{\tx_1, \tx_2\}$ where $\tx_1 \simeq (1, \{k_1\})$ and $\tx_2 \simeq (1, \{k_2\})$. In this case, $\gas^M(T) = \gas^M(\tx_1) + \gas^M(\tx_1) = 1 + 1 = 2 \neq 1 = v(T)$. \qedhere
\end{enumerate}
\end{proof}

Our second practical property requires that gas consumption be efficiently computable, i.e., in polynomial time (\cref{prop:computational_complexity}). A \GCM that does not satisfy this property would be unsuitable for the blockchain context, where gas computation is intended to be a straightforward component.

\begin{property}[Poly-time Computable]\label{prop:computational_complexity}
    There exists a polynomial-time algorithm that takes as input a transaction set $\txset$ and a transaction $\tx$ and outputs 
    $\gas_{\txset \cup \{\tx\}}(\tx)$.
\end{property}

\section{\GCM Proposals}\label{sect:our-gcms}

We now propose multiple \GCM designs. Motivated by the incompatibilities in \cref{thm:impossibility}, these fall into two categories:

\begin{enumerate}[label=(C\arabic*),topsep=0pt,itemsep=-1ex,partopsep=1ex,parsep=1ex,leftmargin=2.45em]
    \item Mechanisms with Easy Gas Estimation, in which each transaction's gas consumption is computed in isolation. Such mechanisms tend to be both straightforward and attractive but can achieve neither Scheduling Monotonicity\footnote{Except for constant \GCMs.} nor Efficiency.
    \item Mechanisms \emph{without} Easy Gas Estimation, which, given a set of transactions $\txset$, rely on $v(T)$ or, more generally, $\left(v(\txset')\right)_{\txset' \subseteq \txset'}$ to holistically calculate gas consumptions for the block $\txset$, requiring knowledge of the entire block. Instead, these mechanisms aim to achieve Scheduling Monotonicity and/or Efficiency.
\end{enumerate}

\subsection{Mechanisms with Easy Gas Estimation}

We already introduced the \emph{current} \GCM (\cref{def:current-gas-comp-mech}).  For the next mechanism, we assume a globally available vector of \emph{positive} weights for the storage keys $(w_k)_{k \in \keyset}$.
For instance, these weights could all be 1. Alternatively, higher-weight storage keys could correspond to storage keys under higher (historical) demand. For our purposes, we only need to assume that the weights are known and do not depend on the block being built. In practice, one could update the weights between blocks to accurately reflect storage key demand (the exact details are not relevant here; see \cref{sec:outlook} for further discussion). Given the weights, we define:

\begin{definition}[Weighted Area \GCM]\label{def:weighted-area-gas-comp-mech}
    Given a set of transactions $\txset$ and a transaction $\tx \in T$ with $\tx \simeq (t, K)$, the \emph{Weighted Area} \GCM computes the amount of gas used by $\tx$ as follows:%
    \begin{equation}
        \gas_{\txset}^\textit{WA}(\tx) := t \cdot \left(1 + \sum_{k \in K} w_k\right)
        \label{def:weighted-area-eq}
    \end{equation}
    Since this does not depend on $T$, we often drop the subscript.
\end{definition}\vspace{2pt}

To gain intuition, it is instructive to consider the former case, where all weights are 1, i.e., the \emph{unweighted area} mechanism, in which case \cref{def:weighted-area-eq} becomes $\gas_{\txset}^\textit{WA}(\tx) = \txtime \cdot \left(1 + |\txkeys|\right)$. 
This mechanism can also be viewed as the addition of two terms: the \emph{current term}, i.e., the gas consumption of $\tx$ under the \emph{current} \GCM, namely $t$, and the \emph{area term}, i.e., the area of a $t \times |K|$ rectangle (which is also how we draw transactions in our diagrams). The \emph{current term} is helpful to ensure no transaction ever costs too little when the weights of all storage keys accessed by a transaction are too close to zero.\footnote{This term has a practical motivation. In reality, transactions will be scheduled on a small finite number of threads. Thus, even if they do not access any high-demand storage key, they occupy space in the schedule proportional to their execution time. Note, moreover, that this term could technically be omitted, and our results would still hold.}  

 The \emph{area term}, on the other hand, which is the main component, intuitively corresponds to ``negated throughput.'' That is, executing the transaction requires holding $|K|$ locks for $t$ time units, during which no other transactions requiring any of those storage keys can execute. In \cref{fig:area}, we illustrate this idea. Transaction $\tx$ prevents the execution of other transactions across its entire storage key set but only utilizes each storage key for a short period. While the current \GCM charges a transaction based solely on its total computation time (i.e., the height of the rectangle), the weighted area \GCM also accounts for the storage keys for which it ``negated throughput'', i.e., the area occupied by the transaction. This occupied area prevents other transactions using the same storage keys from being scheduled in parallel.
\noindent
\begin{minipage}[b]{0.42\textwidth}
    \quad Along the same line of reasoning, to further reinforce why this mechanism is a reasonable choice, consider a transaction set $T$. The sum $\sum_{\tx \in T} t(\tx) \cdot |\txkeys(\tx)|$ represents the ``total area'' of transactions in $\txset$. Dividing this term by the number of storage keys used in $\txset$ provides a lower bound on $v(\txset)$, serving as a rough proxy that does not require knowledge of $T$ when computing individual gas consumptions. Finally, we note that, even in the case of non-unit weights, the term \emph{weighted area} remains meaningful, as the area term still corresponds to the area of the respective rectangle in our diagrams if we give the column of each storage key $k \in \keyset$ a width of $w_k$.
\end{minipage}%
\hfill
\begin{minipage}[b]{0.5\textwidth}
    \centering
    \vspace{-8pt}
    \begin{tikzpicture}[scale=1]

  \draw[-stealth, thick] (0, 0) -- (0, -4.5)  node[midway,above, rotate=180,sloped,inner sep = 20 pt] {\small time};
    \draw[ thick] (-0.1, -3) -- (0, -3)  node[above, rotate=180,sloped,inner sep = 20 pt] {};
    \node[text=gray] at (-0.3,-3)[](R1) {\small 6} ;
  \node[text=gray] at (0.25,0.25)[](R1) {\small $k_1$} ;
  \node[text=gray] at (0.75,0.25)[](R2) {\small $k_2$} ;
  \node[text=gray] at (1.25,0.25)[](R3) {\small $k_3$} ;
  \node[text=gray] at (1.75,0.25)[](R4) {\small $k_4$} ;
  \node[text=gray] at (2.25,0.25)[](R5) {\small $k_5$} ;
  \node[text=gray] at (2.75,0.25)[](R6) {\small $k_6$} ;
  \node[text=gray] at (3.25,0.25)[](R7) {\small $k_7$} ;
  \node[text=gray] at (3.75,0.25)[](R8) {\small $k_8$} ;
  \draw[-stealth, thick] (0, 0) -- (4.5,0)  node[midway,above,sloped,inner sep = 15 pt] {\small storage keys};

    \draw [dotted] (0.5,0) -- (0.5,-4.2);
    \draw [dotted] (1,0) -- (1,-4.2);
    \draw [dotted] (1.5,0) -- (1.5,-4.2);
    \draw [dotted] (2,0) -- (2,-4.2);
    \draw [dotted] (2.5,0) -- (2.5,-4.2);
    \draw [dotted] (3,0) -- (3,-4.2);
    \draw [dotted] (3.5,0) -- (3.5,-4.2);
    \draw [dotted] (4,0) -- (4,-4.2);

    \node at (2,-3.67) [circle,fill,inner sep=1pt]{};
    \node at (2,-3.8) [circle,fill,inner sep=1pt]{};
    \node at (2,-3.93) [circle,fill,inner sep=1pt]{};

    \fill[color1!20] (0.5,-0) rectangle (3.5,-3); 
    \fill[color1!50] (0.5,-0) rectangle (1,-0.5);
    \fill[color1!50] (1,-0.5) rectangle (1.5,-1);
    \fill[color1!50] (1.5,-1) rectangle (2,-1.5);
    \fill[color1!50] (2,-1.5) rectangle (2.5,-2);
    \fill[color1!50] (2.5,-2) rectangle (3,-2.5);
    \fill[color1!50] (3,-2.5) rectangle (3.5,-3);

    \draw[black]  (0.5,-0) rectangle (3.5,-3); 

    \fill[color2!50] (0.5,-3) rectangle (1,-3.5);
    \draw[black]  (0.5,-3) rectangle (1,-3.5);

    \fill[color3!50] (1,-3) rectangle (1.5,-3.5);
    \draw[black]  (1,-3) rectangle (1.5,-3.5);

    \fill[color4!50] (1.5,-3) rectangle (2,-3.5);
    \draw[black]  (1.5,-3) rectangle (2,-3.5);

    \fill[color5!50] (2,-3) rectangle (2.5,-3.5);
    \draw[black]   (2,-3) rectangle (2.5,-3.5);

    \fill[color6!50] (2.5,-3) rectangle (3,-3.5);
    \draw[black]   (2.5,-3) rectangle (3,-3.5);

    \fill[color7!50] (3,-3) rectangle (3.5,-3.5);
    \draw[black]  (3,-3) rectangle (3.5,-3.5);







\end{tikzpicture}\vspace{-10pt}
    \captionof{figure}{Illustration of a worst-case transaction in terms of parallelizability. Transaction $\tx \simeq (6,\{k_2,\dots,k_7\})$ (shown in \fcolorbox{black}{color1!20}{\rule{0pt}{3pt}\rule{3pt}{0pt}}\hspace{1pt}) takes 6 units of time to execute and accesses 6 storage keys. However, the transaction only operates on each storage key for one unit of time each (shown in \fcolorbox{black}{color1!50}{\rule{0pt}{3pt}\rule{3pt}{0pt}}\hspace{1pt}).}
    \label{fig:area}\vspace{0pt}
\end{minipage}

\subsection{Mechanisms \emph{without} Easy Gas Estimation}

We now switch gears towards mechanisms without Easy Gas Estimation that instead aim to achieve Scheduling Monotonicity and/or Efficiency. Intuitively, for a set of transactions $\txset$, such mechanisms should start from $v(T)$ and take each transaction's marginal contribution towards $v(T)$ as its gas consumption. This, however, has to be done carefully, as, e.g., simply looking for each transaction $\tx \in \txset$ at $v(\txset) - v(\txset \setminus \{\tx\})$ does not help. Note that one can easily construct cases where removing any one transaction from $\txset$ does not change $v(\txset)$, so all reported numbers will be 0. However, all these numbers being zero does not mean that no transaction contributes to $v(\txset)$; it just means that we also have to consider removing multiple transactions to see the differences. This gives us the idea to look at the marginal contribution of each transaction $\tx \in \txset$ when added to each possible subset $S\subseteq \txset\setminus\{\tx\}$, and not just to $S = \txset\setminus\{\tx\}$ as before. Formally, we would like to compute $\tx$'s gas consumption as an aggregate of $\left(v(S\cup\{ \tx\} ) - v(S)\right)_{S\subseteq \txset\setminus\{\tx\}}$. We next present two mechanisms based on this idea: the \emph{Shapley} (\cref{def:shapley-gas-comp-mech}) and \emph{Banzhaf} (\cref{def:banzhaf-gas-comp-mech}) \GCMs, inspired by corresponding concepts in cooperative game theory.

\begin{definition}[Shapley \GCM]\label{def:shapley-gas-comp-mech}
    Given a set of transactions $\txset$
    consisting of $|\txset| = n$ transactions
    and a transaction $\tx \in \txset$, the \emph{Shapley} \GCM computes the amount of gas used by $\tx$ as follows:%
    \begin{align}
        \gas_{\txset}^S(\tx) := & \frac{1}{n!} \sum_\sigma \left[ v(P_{\tx}^\sigma \cup \{\tx\}) - v(P_{\tx}^\sigma) \right] \label{shapley:v1}\\
        = & \sum_{S\subseteq \txset\setminus\{\tx\}} \frac { |S|!\cdot (n-|S|-1)!} { n! } \left[ v ( S\cup\{ \tx\} ) - v(S) \right] \label{shapley:v2}.
    \end{align}
    Here, $\sigma$ ranges over the $n!$ possible ways to order the transactions in $\txset$ and $P_{\tx}^\sigma$ denotes the set of transactions that precede $\tx$ in the order $\sigma$. The equality between \cref{shapley:v1,shapley:v2} follows by counting the number of orders $\sigma$ such that $P_{\tx}^\sigma = S$, which there are $|S|! \cdot (n-|S|-1)!$ of.  
\end{definition}\vspace{2pt}

Another way to understand the Shapley \GCM is through the following probabilistic experiment:

\begin{enumerate}
    \item Select an ordering $\sigma$ of the transactions in $\txset$ uniformly at random.
    \item Start with an empty set of transactions and add transactions one by one in the order given by $\sigma$.
    \item Whenever a transaction $\tx$ is added, let $S = P_{\tx}^\sigma$ be the set of transactions just before adding it and compute $\tx$'s \emph{marginal contribution} to the execution time as $v(S \cup \{tx\}) - v(S)$; i.e., by how much did the execution time increase by adding $\tx$ to the current set of transactions.
    \item The gas consumption of $\tx \in \txset$ is the expectation of its marginal contribution across the experiment. 
\end{enumerate}

The reader familiar with cooperative game theory will have already recognized the immediate connection with Shapley values: if we see each transaction $\tx \in \txset$ as a player in a game with valuation function $v : 2^\txset \to \mathbb{R}$, then $\gas_{\txset}^S(\tx)$ is precisely the celebrated \emph{Shapley value} of player $\tx$, more traditionally written $\phi_{\tx}(v)$. One classical property of Shapley values is that, given $v(\varnothing) = 0$, as is the case for us by \ref{sched:empty-set}, their sum equals the valuation of the \emph{grand coalition} $\txset$: $\sum_{\tx \in \txset}\phi_{\tx}(v) = v(\txset)$. The proof is straightforward: for any fixed ordering $\sigma$ of $\txset$, the sum of the marginal contributions of the transactions is a telescoping sum, i.e., except first and last terms, all others appear once positively and once negatively, canceling as a result and leaving us with $v(\txset) - v(\varnothing) = v(\txset)$. Since the sum does not depend on $\sigma$, the same holds when taking the expectation with respect to $\sigma$. This already gives us our first property of the Shapley \GCM, namely Efficiency:

\begin{lemma}\label{lemma:shapley-efficient} The Shapley \GCM satisfies Efficiency (\cref{prop:efficiency}).
\end{lemma}

We postpone studying the Shapley \GCM further until introducing our other mechanisms in the section.

A second \GCM based on the idea that a transaction's gas consumption should be its marginal contribution to the execution time is the \emph{Banzhaf} mechanism:

\begin{definition}[Banzhaf \GCM]\label{def:banzhaf-gas-comp-mech}
    Given a set of transactions $\txset$
    consisting of $|\txset| = n$ transactions
    and a transaction $\tx \in \txset$, the \emph{Banzhaf} \GCM computes the amount of gas used by $\tx$ as follows:%
    \begin{equation*}
        \gas_{\txset}^B(\tx) := \frac{1}{2^{n - 1}}\sum_{S\subseteq \txset\setminus\{\tx\}} \left[ v ( S\cup\{ \tx\} ) - v(S) \right].
    \end{equation*}
\end{definition}\vspace{2pt}

As for the Shapley mechanism, the Banzhaf mechanism can be understood probabilistically, but this time with a separate experiment for each transaction $\tx$. Sample a subset $S$ of transactions other than $\tx$ uniformly at random and compute $\tx$'s marginal contribution to the execution time when added to $S$, namely $v(S \cup \{\tx\}) - v(S)$. The gas consumption of $\tx$ is then its expected marginal contribution to the execution time. The familiar reader will recognize this as the definition of the \emph{Banzhaf power index} $\beta_{\tx}(v)$. For consistency with \emph{Shapley values}, we will instead call these \emph{Banzhaf values}. While more straightforward than the corresponding experiment used in defining Shapley values, the fact that we now have $n$ separate experiments makes the sum of the values no longer well-behaved, losing Efficiency for the Banzhaf mechanism:

\begin{restatable}{lemma}{banzhafnotefficient}\label{lemma:banzhaf-not-efficient} The Banzhaf \GCM does \emph{not} satisfy Efficiency (\cref{prop:efficiency}).
\end{restatable}

We conclude this section by introducing two additional reasonable mechanisms \emph{without} Easy Gas Estimation. These mechanisms are notably more straightforward than the Shapley and Banzhaf \GCMs, as they avoid computing marginal contributions. However, they have other drawbacks that will become more apparent when we begin studying their normative properties alongside the other mechanisms.

\begin{definition}[Time-Proportional Makespan \GCM]\label{def:time-prop-makespan-gas-comp-mech}
    Given a set of transactions $\txset$
    and a transaction $\tx \in \txset$, the \emph{Time-Proportional Makespan (TPM)} \GCM computes the amount of gas used by $\tx$ as follows:%
    \begin{equation*}
        \gas_{\txset}^\textit{TPM}(\tx) := \frac{t(\tx)}{\sum_{\tx' \in T}t(\tx')} \cdot v(T).
    \end{equation*}
\end{definition}\vspace{2pt}

\begin{definition}[Equally-Split Makespan \GCM]\label{def:eq-split-makespan-gas-comp-mech}
    Given a set of transactions $\txset$
    and a transaction $\tx \in \txset$, the \emph{Equally-Split Makespan (ESM)} \GCM computes the amount of gas used by $\tx$ as follows:%
    \begin{equation*}
        \gas_{\txset}^\textit{ESM}(\tx) := \frac{v(\txset)}{|\txset|}.
    \end{equation*}
\end{definition}\vspace{2pt}

We also consider the following rather pathological mechanism because of the combination of properties it turns out to satisfy (that none of our other mechanisms do):

\begin{definition}[Exponentially-Split Makespan \GCM]\label{def:exp-split-makespan-gas-comp-mech}
    Given a set of transactions $\txset$
    and a transaction $\tx \in \txset$, the \emph{Exponentially-Split Makespan (XSM)} \GCM computes the amount of gas used by $\tx$ as follows:%
    \begin{equation*}
        \gas_{\txset}^\textit{XSM}(\tx) := \frac{v(\txset)}{3^{|\txset|}}.
    \end{equation*}
\end{definition}\vspace{2pt}

\section{Analysis of Our GCMs} \label{sect:gcm-analysis}

In this section, we provide a detailed analysis of the normative properties of our \GCMs. Our results are summarized in \cref{tbl:results}.

\begin{table}[h!]
\centering
\resizebox{\columnwidth}{!}{%
\begin{tabular}{lccccccc}
\toprule
Property & Current & W.~Area & Shapley & Banzhaf & TPM & ESM  & XSM   \\ 
\midrule
Storage Key Monotonicity (\cref{prop:resource_monotonicity}) & \cellcolor{twoyellow}$=$  &  \cellcolor{twodarkgreen}$<$  & \cellcolor{twogreen}$\leq$&\cellcolor{twogreen}$\leq$ & \cellcolor{twogreen}$\leq$  & \cellcolor{twogreen}$\leq$ & \cellcolor{twogreen}$\leq$\\
Time Monotonicity (\cref{prop:time_monotonicity}) & \cellcolor{twodarkgreen}$<$&  \cellcolor{twodarkgreen}$<$  & \cellcolor{twodarkgreen}$<$ &\cellcolor{twodarkgreen}$<$ &\cellcolor{twodarkgreen}$<$  & \cellcolor{twogreen}$\leq$ & \cellcolor{twogreen}$\leq$ \\ 
Storage Key-Time Monotonicity (\cref{prop:resource_time_monotonicity}) & \cellcolor{twogreen}$\leq$  &  \cellcolor{twodarkgreen}$<$  & \cellcolor{twogreen}$\leq$ &\cellcolor{twogreen}$\leq$ &  \cellcolor{twogreen}$\leq$ & \cellcolor{twogreen}$\leq$ & \cellcolor{twogreen}$\leq$ \\ 
Set Inclusion  (\cref{prop:se\txtime_inclusion})& \cellcolor{twodarkgreen}$<$  &  \cellcolor{twodarkgreen}$<$  & \cellcolor{twored}\tikzxmark  &\cellcolor{twored}\tikzxmark  & \cellcolor{twogreen}$\leq$   & \cellcolor{twogreen}$\leq$ & \cellcolor{twored}\tikzxmark \\ 
Transaction Bundling  (\cref{prop:tx_bundling}) &  \cellcolor{twoyellow}$=$ &\cellcolor{twogreen}$\leq$  & \cellcolor{twored}\tikzxmark  &\cellcolor{twogreen}$\leq$ & \cellcolor{twogreen}$\leq$  & \cellcolor{twored}\tikzxmark &\cellcolor{twodarkgreen}$<$\\
Scheduling Monotonicity (\cref{prop:schedule_monotonicity})& \cellcolor{twored}\tikzxmark & \cellcolor{twored}\tikzxmark & \cellcolor{twored}\tikzxmark  & \cellcolor{twored}\tikzxmark & \cellcolor{twored}\tikzxmark &  \cellcolor{twodarkgreen}$<$ &   \cellcolor{twodarkgreen}$<$ \\ 
Efficiency (\cref{prop:efficiency})&\cellcolor{twored}\tikzxmark    &\cellcolor{twored}\tikzxmark  &\cellcolor{twodarkgreen}\tikzcmark & \cellcolor{twored}\tikzxmark & \cellcolor{twodarkgreen}\tikzcmark  & \cellcolor{twodarkgreen}\tikzcmark &\cellcolor{twored}\tikzxmark \\ 
Easy Gas Estimation (\cref{prop:Composability})& \cellcolor{twodarkgreen}\tikzcmark &\cellcolor{twodarkgreen}\tikzcmark & \cellcolor{twored}\tikzxmark &\cellcolor{twored}\tikzxmark& \cellcolor{twored}\tikzxmark &\cellcolor{twored}\tikzxmark& \cellcolor{twored}\tikzxmark \\ 
Poly-time Computable (\cref{prop:computational_complexity})& \cellcolor{twodarkgreen}\tikzcmark & \cellcolor{twodarkgreen}\tikzcmark & \cellcolor{twored}$S(v)$ &\cellcolor{twored}$B(v)$& \cellcolor{twoyellow}$v$  &\cellcolor{twoyellow}$v$  & \cellcolor{twoyellow}$v$ \\ 

 \bottomrule
\end{tabular}}\vspace{4pt}

\caption{Comparison of \GCMs based on their adherence to the defined properties. \(\,<\,\) indicates that the mechanism \emph{strictly} satisfies the property,
\(\,=\,\) indicates \emph{trivial} satisfaction (by equality),
and \(\,\leq\,\) indicates satisfaction (not necessarily strict).
A ``\(\checkmark\)'' means the property is satisfied; 
``\(\times\)'' means it is not satisfied.
For computational complexity, 
``\(v\)'' means the mechanism is as hard to compute as \(v(\cdot)\) itself,
``\(S(v)\)'' means it is as hard as computing Shapley values for \(v\),
and ``\(B(v)\)'' means it is as hard as computing Banzhaf values for \(v\).
 }
\label{tbl:results}\vspace{-20pt}
\end{table}

\subsection{Mechanisms with Easy Gas Estimation}

In this section, we analyze the current and Weighted Area \GCMs. By definition, both satisfy Easy Gas Estimation (\cref{prop:Composability}) and are Poly-time Computable (\cref{prop:computational_complexity}). Since they satisfy Easy Gas Estimation but are not constant \GCMs, \cref{thm:impossibility} implies that \emph{neither} satisfies Scheduling Monotonicity (\cref{prop:schedule_monotonicity}) nor Efficiency (\cref{prop:efficiency}). Furthermore, both mechanisms satisfy strict Time Monotonicity (\cref{prop:time_monotonicity}): increasing the time $t$ of a transaction strictly increases its gas consumption. Similarly, both satisfy strict Set Inclusion (\cref{prop:set_inclusion}): since transactions' gas consumptions are strictly positive, a strict superset of a given set of transactions consumes strictly more gas. For the remaining three properties, the two mechanisms behave slightly differently: 

\emph{Storage Key Monotonicity (\cref{prop:resource_monotonicity}).} The current mechanism ignores the set of storage keys $K$ that a transaction accesses, so replacing $K$ with a strict superset of it does not change the transaction's gas consumption. Therefore, the current mechanism satisfies Storage Key Monotonicity \emph{with equality}. On the other hand, the Weighted Area mechanism adds an extra term of $t \cdot w_k > 0$ to the gas consumption for each additional storage key $k$ added to $K$, so it satisfies strict Storage Key Monotonicity.

\emph{Storage Key-Time Monotonicity (\cref{prop:resource_time_monotonicity}).} From the above results on whether our two mechanisms satisfy Storage Key Monotonicity and Time Monotonicity, \cref{lemma:prop-1-2-3-equivalence,lemma:prop-1-2-3-equivalence-strict} allow us to conclude that the current mechanism satisfies Storage Key-Time Monotonicity, while the Weighted Area mechanism satisfies it strictly.

\emph{Transaction Bundling (\cref{prop:tx_bundling}).} Concatenating two transactions with times $t_1$ and $t_2$ results in a transaction with time $t_1 + t_2$. Since the current mechanism equates gas consumption with time, the bundled transaction has the same gas consumption as the two individual transactions combined. This implies that the current mechanism satisfies Transaction Bundling \emph{with equality}. Finally, the Weighted Area mechanism satisfies Transaction Bundling, as shown below. If we restrict ourselves to bundling transactions with different storage key sets, the property is strictly satisfied.

\begin{lemma} The Weighted Area \GCM satisfies Transaction Bundling (\cref{prop:tx_bundling}). If we only consider bundling transactions with different storage key sets, the property is strictly satisfied.
\end{lemma}
\begin{proof} Consider two transactions $\tx_1 \simeq (t_1, K_1)$ and $\tx_2 \simeq (t_2, K_2)$. Let $\tx_3 \simeq (t_1 + t_2, K_1 \cup K_2)$ be a transaction consisting of the concatenation of $\tx_1$ and $\tx_2$. We want to show that $\gas^\textit{WA}(\tx_1)+\gas^\textit{WA}(\tx_2) \leq \gas^\textit{WA}(\tx_3)$. By definition, this amounts to:

\begin{equation}
    t_1 \cdot \left(1 + \sum_{k \in K_1} w_k\right) + 
    t_2 \cdot \left(1 + \sum_{k \in K_2} w_k\right) \leq 
    (t_1 + t_2) \cdot \left(1 + \sum_{k \in K_1 \cup K_2} w_k\right)  
\end{equation}

Which is true because $\sum_{k \in K_i} w_k \leq \sum_{k \in K_1 \cup K_2} w_k$ for $i \in \{1, 2\}$. Because the weights are strictly positive, equality occurs if and only if $K_1 = K_1 \cup K_2$ and $K_2 = K_1 \cup K_2$; i.e., $K_1 = K_2$. Hence, the property is satisfied strictly if we restrict bundling transactions to cases where $K_1 \neq K_2$.
\end{proof}

\subsection{Mechanisms \emph{without} Easy Gas Estimation}

In this section, we analyze the Shapley, Banzhaf, TPM, ESM, and XSM \GCMs. By definition, all of them require knowledge of $T$ to compute $\gas_T(\tx)$, so they do not satisfy Easy Gas Estimation (\cref{prop:Composability}). Next, we examine each of the remaining properties individually and analyze whether they hold for our five mechanisms.

\emph{Poly-time Computable (\cref{prop:computational_complexity}).} The TPM, ESM, and XSM \GCMs are all Poly-time Computable whenever determining the makespan of a given set of transactions under the chosen scheduler is feasible in polynomial time, i.e., when $v(T)$ can be computed in polynomial time.\footnote{Notably, this is essentially never true for non-trivial makespan minimization problems. In particular, if we restrict our attention to the case of unit-length transactions (i.e., with $t = 1$) and infinitely many threads $n = \infty$, then checking for the existence of a schedule with makespan $c$ corresponds to checking whether the intersection graph of the transactions (i.e., where edges correspond to transactions with intersecting storage key sets) is $c$-colorable. For $c = 3$, this is well-known to be NP-complete, but this result is for general graphs. However, there is a straightforward way to model any graph $G = (V, E)$ as an intersection graph of transactions: vertices $v \in V$ are transactions and for every edge $(u, v) \in E$, create a new storage key $k$ and add it to the storage key sets of transactions $u$ and $v$.} In fact, the problems are all equally difficult to computing such makespans. In contrast, computing gas consumptions under the Shapley and Banzhaf \GCMs hits another hurdle: it is no longer enough to be able to efficiently compute $v(T)$, but instead, one needs to be able to compute the Shapley or Banzhaf values of the transactions, involving aggregating over $(v(T'))_{T' \subseteq T}$.
Hence, computing Shapley and Banzhaf values is, in general, NP-hard and \#P-complete \cite{shapley_banzhaf_1,shapley_banzhaf_2,shapley_banzhaf_3,shapley_banzhaf_4,shapley_banzhaf_5}. Moreover, no deterministic polynomial-time algorithm can approximate the Shapley values within a constant factor unless P = NP. Using randomization could, in principle, circumvent this: to approximate an average consisting of exponentially many terms, sample polynomially many uniformly at random, and take their average. One could imagine implementing this in a blockchain via a VRF or similar mechanism to ensure unbiased randomness. 
Still, this comes with a notable downside: the accuracy of the computed gas consumption may be hard to predict in advance and might vary wildly.
Moreover, even with randomization, one could still run into issues computing the sampled terms, which is as hard as evaluating $v$ a constant number of times. See \cite{cooperative_book_edith} for an ampler discussion of computing Shapley and Banzhaf values. Note that the previous results mostly pertain to functions $v$ of a different shape than ours (i.e., not related to scheduling). We have not attempted to show that hardness is retained in our context, but expect this to be true. Note still that because the Shapley \GCM is Efficient, it can be used to compute $v(T)$ by adding up the gas consumptions of all transactions in $T$, meaning that the Shapley \GCM is at least as difficult to compute as $v(T)$. However, this simple reduction no longer works for the Banzhaf \GCM.

\emph{Efficiency (\cref{prop:efficiency}).} We have already seen that the Shapley \GCM is Efficient (\cref{lemma:shapley-efficient}) while the Banzhaf \GCM is not (\cref{lemma:banzhaf-not-efficient}). Moreover, by adding up the gas consumption, one can immediately see from the definitions that TPM and ESM are Efficient, while XSM is not. We take this occasion to also note that any \GCM can be made Efficient by appropriately scaling its output. Most notably, applying this to the Banzhaf \GCM yields a mechanism based on the well-known \emph{normalized Banzhaf values}, though the resulting mechanism otherwise seems rather poorly behaved.

\emph{Scheduling Monotonicity (\cref{prop:schedule_monotonicity}).} The ESM and XSM \GCMs can be easily seen to strictly satisfy Scheduling Monotonicity. This is because, under both mechanisms, given a set of transactions $T$, the gas consumption of any transaction in $T$ is computed as $f(|T|) \cdot v(T)$, where $f$ is either $\frac{1}{x}$ or $\frac{1}{3^x}$. Notably, the first factor depends only on $|T|$, so if a transaction in $T$ were modified in a way that increases the makespan, this increase would also be reflected proportionally in its gas consumption. In contrast, the Shapley, Banzhaf, and TPM \GCMs do \emph{not} satisfy Scheduling Monotonicity (proven in \cref{lem:shapleyschedule} below). This may be particularly surprising for the Shapley and Banzhaf \GCMs, as they were specifically designed to account for marginal increases in makespan. The catch is that Scheduling Monotonicity considers replacing a transaction $\tx$ with another one leading to an increase in makespan. However, some elements in $\left(v(S\cup\{ \tx\} ) - v(S)\right)_{S\subseteq \txset\setminus\{\tx\}}$ may still decrease as a result (except for $S = \txset\setminus\{\tx\}$, for which we assumed an increase). Because both the Shapley and Banzhaf values of $\tx$ take a weighted average over these values, the average might still decrease, which is what happens in the example in the lemma below.

\begin{lemma}\label{lem:shapleyschedule} The
Shapley, Banzhaf and TPM \GCMs do \emph{not} satisfy Scheduling Monotonicity (\cref{prop:schedule_monotonicity}).
\end{lemma}

\begin{proof} Consider the set of transactions \(T=\{\tx_1,\tx_2\}\) and two additional transactions $\tx_3$ and $\tx_4$ such that $\tx_1\simeq(1,\{k_1\})$, $\tx_2\simeq(3,\{k_2\})$, $\tx_3\simeq(2,\{k_1\})$, and $\tx_4\simeq(1,\{k_2\})$. Then, for any number of threads $n \geq 2$ and the optimal scheduler we have $v(\txset \cup \{\tx_3\}) =3 < 4 = v(\txset\cup\{\tx_4\})$. However:
\begin{align*}
    \gas_{\txset\cup\{\tx_3\}}^S(\tx_3)=\frac{2+2+2+0+0+0}{6} = 1 >& \frac{5}{6} = \frac{1+1+1+1+1+0}{6} = \gas_{\txset\cup\{\tx_4\}}^S(\tx_4)\\
    \gas_{\txset\cup\{\tx_3\}}^B(\tx_3)=\frac{2+2+0+0}{4} = 1 > & \frac{3}{4} = \frac{1+0+1+1}{4} = \gas_{\txset\cup\{\tx_4\}}^B(\tx_4) \\
    \gas_{\txset\cup\{\tx_3\}}^\textit{TPM}(\tx_3)=\frac{2}{1+3+2}\cdot 3 = 1 > & \frac{4}{5} = \frac{1}{1+3+1} \cdot 4 = \gas_{\txset\cup\{\tx_4\}}^\textit{TPM}(\tx_4).
\end{align*}

So, the Shapley, Banzhaf, and TPM \GCMs all violate Scheduling Monotonicity in this case. 
\end{proof}

\emph{Transaction Bundling  (\cref{prop:tx_bundling}).} We find that the Banzhaf, TPM and XSM \GCMs satisfy Transaction Bundling, only XSM satisfying it strictly, while the Shapley and ESM \GCMs do not satisfy the property. We prove these facts in the following 5 lemmas.

\begin{restatable}{lemma}{lemmashapleynotxbundling} The Shapley \GCM does \emph{not} satisfy Transaction Bundling (\cref{prop:tx_bundling}).
\end{restatable}

\begin{restatable}{lemma}{lemmabanzhafhastxbundl} The Banzhaf \GCM satisfies Transaction Bundling (\cref{prop:tx_bundling}).
\end{restatable}

\begin{restatable}{lemma}{lemmatpmhasbundling} The TPM \GCM satisfies Transaction Bundling (\cref{prop:tx_bundling}).
\end{restatable}

\begin{restatable}{lemma}{lemmaesmnobundling} The ESM \GCM does \emph{not} satisfy Transaction Bundling (\cref{prop:tx_bundling}).
\end{restatable}

\begin{restatable}{lemma}{lemmaxsmhassricttxbundling} The XSM \GCM strictly satisfies Transaction Bundling (\cref{prop:tx_bundling}).
\end{restatable}

\emph{Set Inclusion  (\cref{prop:set_inclusion}).} We find that the TPM and ESM \GCMs satisfy Set Inclusion, while the Shapley, Banzhaf and XSM \GCMs do not satisfy the property. We prove these facts in the following 5 lemmas.

\begin{restatable}{lemma}{lemmashapleynotsetincl} The Shapley \GCM does \emph{not} satisfy Set Inclusion (\cref{prop:set_inclusion}).
\end{restatable}

\begin{restatable}{lemma}{lemmabanzhafnosetincl} The Banzhaf \GCM does \emph{not} satisfy Set Inclusion (\cref{prop:set_inclusion}).
\end{restatable}

\begin{restatable}{lemma}{lemmatpmhassetincl} The TPM \GCM satisfies Set Inclusion (\cref{prop:set_inclusion}).
\end{restatable}

\begin{restatable}{lemma}{lemmaesmhassetincl} The ESM \GCM satisfies Set Inclusion (\cref{prop:set_inclusion}).
\end{restatable}

\begin{restatable}{lemma}{lemmaxsmnosetinclusion} The XSM \GCM does \emph{not} satisfy Set Inclusion (\cref{prop:set_inclusion}).
\end{restatable}

\emph{Storage Key-Time Monotonicity (\cref{prop:resource_time_monotonicity}).} All five mechanisms satisfy this property. For the ESM and XSM \GCMs, this is an immediate consequence of property \ref{sched:monot-t-R}. For the Shapley and Banzhaf \GCMs, one can see this by recalling that they compute the gas consumption of a transaction $\tx \in T$ as a weighted average over $\left(v(S\cup\{ \tx\} ) - v(S)\right)_{S\subseteq \txset\setminus\{\tx\}}$. Hence, by property \ref{sched:monot-t-R}, when $\tx$ is replaced with some $\tx' \gtrsim \tx$, no term in the previous decreases, so their weighted average also does not decrease. Finally, for the TPM \GCM, we show this in the lemma after the next paragraph*.

\emph{Storage Key Monotonicity (\cref{prop:resource_monotonicity}) and Time Monotonicity (\cref{prop:time_monotonicity}).} Because all five mechanisms satisfy Storage Key-Time Monotonicity, by \cref{lemma:prop-1-2-3-equivalence}, they also all satisfy Storage Key Monotonicity and Time Monotonicity. Out of the five mechanisms, the Shapley, Banzhaf, and TPM \GCMs satisfy the property strictly. For the first two, this is because when $t(\tx)$ increases, at least one term in $\left(v(S\cup\{ \tx\} ) - v(S)\right)_{S\subseteq \txset\setminus\{\tx\}}$ strictly increases, namely the term for $S = \varnothing$, while no terms decrease by property \ref{sched:monot-t-R}. Last, for the TPM \GCM, we show this in the lemma below.

\begin{restatable}{lemma}{lemmatpmrestimesensandtimesens} The TPM \GCM satisfies Storage Key-Time Monotonicity (\cref{prop:resource_time_monotonicity}) and strict Time Monotonicity (\cref{prop:time_monotonicity}). 
\end{restatable}

\section{Towards a Fee Market for Parallel Execution}\label{sec:outlook}

Armed with an understanding of the trade-offs between desirable properties in a \GCM, in particular the impossibility of satisfying all properties within a single mechanism, and informed by our analysis of various candidate mechanisms and their properties, we propose two mechanisms for practical implementations of a fee market supporting parallel execution. Each represents one side of the design spectrum: one drawn from the class of mechanisms with Easy Gas Estimation, and the other from the class without it.

The advantage of adopting a mechanism with Easy Gas Estimation is that satisfying this property ensures that each transaction consumes a fixed amount of gas, regardless of other transactions in the same block. From the perspective of currently deployed TFMs (e.g., EIP‑1559), which process transactions with fixed sizes, nothing changes. Thus, the existing properties of these mechanisms remain intact. On a high level, the key properties we strive for in a TFM are incentive compatibility for both block producers and users, welfare optimality, and collusion resistance. However, no TFM can achieve all these properties simultaneously~\cite{chung2023foundations,chung2024collusion,gafni2024barriers}. Importantly, when composing a \GCM that satisfies Easy Gas Estimation with a TFM of choice, the level of sophistication required from users in their bidding strategy does not increase, unlike in currently deployed fee markets for parallel execution~\cite{lostin2025truth,mueller2025}.

On the other hand, foregoing Easy Gas Estimation makes it possible to price transactions according to the load they impose on the network relative to the other transactions in the block, rather than evaluating each one in isolation. Mechanisms in this class can align the total gas charged in a block with the actual execution cost of that block, enabling resource pricing at the block level rather than only at the transaction level.
Recent discussions in Ethereum research have explored similar designs under the umbrella of block-level fee markets~\cite{ethresearch_blockfee2025}. If the protocol is already moving toward block-level metering for other resources, such as data availability or storage contention, extending this approach to execution costs could enable direct efficiency gains from parallel execution that are not achievable with transaction-level pricing alone.

For the class of mechanisms with Easy Gas Estimation, the natural choice is the Weighted Area \GCM. Given that it satisfies Easy Gas Estimation, it achieves the most additional properties one could hope for in a (non-constant) \GCM (see \cref{thm:impossibility}).

For the class of mechanisms without Easy Gas Estimation, we propose the Time-Proportional Makespan \GCM as a concrete candidate. By allocating gas in proportion to each transaction’s execution time relative to the block’s total makespan, it provides a simple way to align gas consumption with how transactions constrain parallel execution within a block.

The decision between these approaches ultimately depends on protocol-level priorities: maintaining user-side simplicity and compatibility with existing TFMs through Easy Gas Estimation, which also makes implementation simpler by reusing current transaction-level pricing infrastructure, or pursuing block-level efficiency and more accurate resource pricing at the cost of giving up the previous benefits.

\section{Related Work}

\subsection{Transaction Fee Mechanisms}

There is extensive research on blockchain fee markets, with a particular focus on Ethereum and Bitcoin. Early studies primarily examined Bitcoin, exploring monopolistic pricing mechanisms~\cite{lavi2022redesigning,yao2020incentive}. More recent contributions to this field include~\cite{nisan2023serial,gafni2022greedy,penna2024serial}. Unlike these works, our study concerns measuring storage key usage on a blockchain with client-side parallel execution, rather than focusing on pricing.

The TFM design framework was introduced by Roughgarden~\cite{roughgarden2020transaction,roughgarden2024transaction}. Roughgarden's analysis of the EIP-1559 mechanism~\cite{eip1559spec} initiated an active line of research on TFMs. Chung and Shi~\cite{chung2023foundations} demonstrated that no TFM can be ideal --- meaning it cannot simultaneously be incentive-compatible for users and block producers while also being resistant to collusion between the two. This conclusion holds even for weaker definitions of collusion resilience, as shown by Chung et al.~\cite{chung2024collusion} and Gafni and Yaish \cite{gafni2024barriers}. Finally, attempts to address these limitations using cryptographic techniques \cite{shi2022can,wu2023maximizing} have made progress in overcoming certain impossibilities, while other attempts relax the desiderata~\cite{gafni2024discrete}. However, designing an ideal TFM still remains out of reach. While these studies examine the limitations of TFMs, our focus is on GCMs for parallel execution and how to integrate them with a TFM.

A related body of work examines the dynamics of TFMs over multiple blocks, particularly focusing on the base fee in EIP-1559. Leonardos et al.~\cite{leonardos2021dynamical,leonardos2023optimality} demonstrate that the stability of the base fee depends on the adjustment parameter, with short-term volatility but long-term block size stability. Reijsbergen et al.~\cite{reijsbergen2021transaction} suggest using an adaptive adjustment parameter to mitigate block size fluctuations, while Ferreira et al.~\cite{ferreira2021dynamic} highlight user experience issues caused by bounded base fee oscillations. Additionally, Hougaard and Pourpouneh~\cite{hougaard2023farsighted} and Azouvi et al.~\cite{azouvi2023base} reveal that the base fee can be manipulated by non-myopic miners.

Given the discussion surrounding multi-dimensional fees in Ethereum~\cite{multidimensional_eip1559,buterin2024multidim} and the deployment of EIP-4844~\cite{eip4844} (a first step towards a multi-dimensional fee market on Ethereum), a recent line of work explores multi-dimensional fee markets, focusing on efficient pricing mechanisms and their optimality. This work is further refined by Diamandis et al.~\cite{diamandis2023designing}, who design and analyze multi-dimensional blockchain fee markets to align incentives and improve network performance. Building on this, Angeris et al.~\cite{angeris2024multidimensional} prove that such fee markets are nearly optimal, with efficiency improving over time even under adversarial conditions. Multidimensional fee markets are closely related to fee markets designed for parallel execution. In particular, in the weighted area \GCM, the weights can be interpreted as fees within a multidimensional fee market. Unlike previous literature on multidimensional fee markets, we focus on parallelization, introduce desirable properties, and evaluate how various mechanisms perform.

Further extensions of TFMs have emerged. Bahrani et al.~\cite{bahrani2024transaction} consider TFMs in the presence of maximal extractable value (MEV), i.e., value extractable by the block producer. Further, Wang et al.~\cite{wang2024mechanism} design a fee mechanism for proof networks, whereas Bahrani et al.~\cite{bahrani2024resonance} introduce a transaction fee mechanism for heterogeneous computation. Our work most closely relates to the latter, but, in contrast, our chosen approach is closer to multidimensional fee markets, trading complexity for the block producer for stronger incentive compatibility for the user.  

Local fee markets have recently been a topic of discussion in the blockchain space~\cite{eclipse2024local,diamandis2024toward,lostin2025truth,keyneom2024local}. The core idea is that transactions interacting with highly contested states incur higher fees, while those involving non-contested states pay lower fees. However, discussions on local fee markets have largely remained high-level, without a precise characterization of the desired properties beyond this general goal. Moreover, currently implemented local fee markets~\cite{lostin2025truth} require significant user sophistication to set fees appropriately. In this work, we formalize the desiderata for fee markets in the context of parallel execution and identify the weighted area \GCM as a promising candidate. One key advantage is its compatibility with a TFM, enabling simple fee estimation for users.

\subsection{Parallel Execution}

Blockchain concurrency has been a focal point in an active line of research. In particular, numerous efforts have aimed to enable parallel transaction processing through speculative execution~\cite{sergey2017concurrent,zhang2018enabling,amiri2019parblockchain,dickerson2017adding,anjana2022optsmart,gelashvili2022block,chen2021forerunner,saraph2019empirical}. Note that speculative execution is already deployed by multiple blockchains~\cite{aptos,sei_protocol,monad}. Static analysis has also been employed to identify parallelizable transactions, though it cannot completely eliminate inherent dependencies~\cite{pirlea2021practical,murgia2021theory}. Similarly, Neiheiser et al.~\cite{neiheiser2024pythia} demonstrate how parallel execution can assist struggling nodes in catching up. While these works are orthogonal to ours, they highlight the overhead of parallel execution when there is no advance knowledge about a transaction's state accesses.

Further, Saraph and Herlihy \cite{saraph2019empirical} and Heimbach et al.~\cite{heimbach2023defi} have evaluated the parallelization potential of the Ethereum workload. The latter demonstrates that a speedup of approximately fivefold is achievable, assuming state accesses are known in advance. Additionally, Solana~\cite{solana} and Sui~\cite{sui} already perform parallel execution with advance knowledge of state accesses. However, in practice, state accesses are not known beforehand on many blockchains such as Ethereum. There, less than 2\% of transactions disclose them proactively, as shown by Heimbach et al.~\cite{heimbach2023dissecting}, due to a lack of incentives. In this work, we aim to take a step toward unlocking the parallelization potential by designing a TFM that supports parallel execution. This mechanism relies on the disclosure of state accesses as done in Solana~\cite{solana} and Sui~\cite{sui}.

\subsection{Cloud Computing Pricing and Parallel Resource Allocation}

While our primary focus is on transaction fee mechanisms for blockchain networks, parallels can be drawn to resource allocation and pricing in cloud computing. Several papers study pricing schemes for cloud computing resources in a monopolistic setting, where a single provider sells access to multiple resources and sets prices to optimize revenue~\cite{10.1145/2479942.2479944}. In these models, the provider posts prices for resource usage, users select resource bundles to maximize their individual surplus, and the provider earns revenue from the resulting allocation. Mechanisms such as CloudPack~\cite{inproceedings} extend this approach by incorporating workload flexibility, using concepts like Shapley values for fair cost distribution among colocated jobs. Similarly, works such as~\cite{Pei_2024} address resource contention through dynamic pricing strategies that adjust costs based on runtime slowdowns.

While structurally similar, the blockchain context introduces an additional constraint. In addition to objectives like maximizing provider revenue, it emphasizes maximizing throughput by ensuring transactions are processed efficiently, minimizing their collective execution time within a block. Moreover, blockchain transactions explicitly declare access to distinct storage keys, introducing contention through overlapping access patterns. This adds complexity compared to cloud models, which typically abstract resource demands as scalar quantities and address contention at a more aggregated level.

Capturing the effects of storage key-level conflicts on parallel execution in blockchains requires fee computation mechanisms that are sensitive to these fine-grained interactions. Thus, while conceptually addressing similar challenges, the blockchain setting imposes unique constraints that motivate the need for specialized GCMs as explored in this work.

\section{Conclusion}

In this work, we took a step towards creating a fee market that meets the demands of parallel execution environments while also upholding the properties we want from a TFM. 

Recently, the idea of local fee markets has been proposed for blockchains that support parallel execution. However, to the best of our knowledge, before this work, the demands on these fee markets have only been outlined at a very high level, and the markets that have been implemented are not ideal yet, e.g., they require high levels of sophistication from users when bidding. 

In this work, we addressed this gap by introducing a framework with two key components: a GCM, which measures the execution-related load a transaction imposes on the network in units of gas, and a TFM, which determines the cost associated with each unit of gas. We then formalized the desired properties for the GCM in such a fee market. After outlining these desiderata, we evaluated various mechanisms against them and identified two strong candidates through this analysis: the \textit{weighted area} \GCM for the class of mechanisms with Easy Gas Estimation, and the \textit{time-proportional makespan} \GCM for the class without it.

Setting the right incentives in fee markets for parallel execution is crucial to unlocking the full potential of execution layer parallelization, and we hope that our work contributes to the development of fee markets capable of meeting the demands of such environments.



\bibliography{lipics-v2021-sample-article}

\begin{thebibliography}{10}

\bibitem{amiri2019parblockchain}
Mohammad~Javad Amiri, Divyakant Agrawal, and Amr El~Abbadi.
\newblock {Parblockchain: Leveraging Transaction Parallelism in Permissioned Blockchain Systems}.
\newblock In {\em 2019 IEEE 39th International Conference on Distributed Computing Systems (ICDCS)}, pages 1337--1347. IEEE, 2019.

\bibitem{angeris2024multidimensional}
Guillermo Angeris, Theo Diamandis, and Ciamac Moallemi.
\newblock {Multidimensional Blockchain Fees are (Essentially) Optimal}.
\newblock {\em arXiv preprint arXiv:2402.08661}, 2024.

\bibitem{anjana2022optsmart}
Parwat~Singh Anjana, Sweta Kumari, Sathya Peri, Sachin Rathor, and Archit Somani.
\newblock Optsmart: A space efficient optimistic concurrent execution of smart contracts.
\newblock {\em Distributed and Parallel Databases}, pages 1--53, 2022.

\bibitem{aptos}
{Aptos Labs}.
\newblock Aptos, 2023.
\newblock Accessed: 2025-01-23.
\newblock URL: \url{https://www.aptoslabs.com}.

\bibitem{azouvi2023base}
Sarah Azouvi, Guy Goren, Lioba Heimbach, and Alexander Hicks.
\newblock {Base Fee Manipulation in Ethereum’s EIP-1559 Transaction Fee Mechanism}.
\newblock In {\em 37th International Symposium on Distributed Computing}, 2023.

\bibitem{bahrani2024resonance}
Maryam Bahrani and Naveen Durvasula.
\newblock {Resonance: Transaction Fees for Heterogeneous Computation}.
\newblock {\em arXiv preprint arXiv:2411.11789}, 2024.

\bibitem{bahrani2024transaction}
Maryam Bahrani, Pranav Garimidi, and Tim Roughgarden.
\newblock {Transaction Fee Mechanism Design with Active Block Producers}.
\newblock In {\em International Conference on Financial Cryptography and Data Security}, pages 85--90. Springer, 2024.

\bibitem{multidimensional_eip1559}
Vitalik Buterin.
\newblock {Multidimensional EIP-1559}.
\newblock \url{https://ethresear.ch/t/multidimensional-eip-1559/11651}, 2022.
\newblock Accessed: 2025-01-23.

\bibitem{buterin2024multidim}
Vitalik Buterin.
\newblock {Multidimensional Pricing for EIP-1559}.
\newblock \url{https://vitalik.eth.limo/general/2024/05/09/multidim.html}, May 2024.
\newblock Accessed: 2025-01-23.

\bibitem{ethresearch_blockfee2025}
Vitalik Buterin.
\newblock Block-level fee markets: Four easy pieces.
\newblock \url{https://ethresear.ch/t/block-level-fee-markets-four-easy-pieces/21448}, 2025.
\newblock Accessed: 2025-08-04.

\bibitem{eip1559spec}
Vitalik Buterin, Eric Conner, Rick Dudley, Matthew Slipper, Ian Norden, and Abdelhamid Bakhta.
\newblock Eip-1559: Fee market change for eth 1.0 chain.
\newblock \url{https://github.com/ethereum/EIPs/blob/master/EIPS/eip-1559.md}, 2024.
\newblock Accessed: 2025-01-23.

\bibitem{eip4844}
Vitalik Buterin, Dankrad Feist, Diederik Loerakker, George Kadianakis, Matt Garnett, Mofi Taiwo, and Ansgar Dietrichs.
\newblock {EIP-4844: Shard Blob Transactions}.
\newblock \url{https://eips.ethereum.org/EIPS/eip-4844}, 2022.
\newblock Accessed: 2025-01-23.

\bibitem{eip2930}
Vitalik Buterin and Martin Swende.
\newblock Eip-2930: Optional access lists, October 2020.
\newblock Accessed: 2025-01-22.
\newblock URL: \url{https://eips.ethereum.org/EIPS/eip-2930}.

\bibitem{cooperative_book_edith}
Georgios Chalkiadakis, Edith Elkind, and Michael Wooldridge.
\newblock {\em Weighted Voting Games}, pages 49--70.
\newblock Springer International Publishing, Cham, 2012.
\newblock \href {https://doi.org/10.1007/978-3-031-01558-8_5} {\path{doi:10.1007/978-3-031-01558-8_5}}.

\bibitem{chen2021forerunner}
Yang Chen, Zhongxin Guo, Runhuai Li, Shuo Chen, Lidong Zhou, Yajin Zhou, and Xian Zhang.
\newblock {Forerunner: Constraint-Based Speculative Transaction Execution for Ethereum}.
\newblock In {\em Proceedings of the ACM SIGOPS 28th Symposium on Operating Systems Principles}, pages 570--587, 2021.

\bibitem{chung2024collusion}
Hao Chung, Tim Roughgarden, and Elaine Shi.
\newblock {Collusion-Resilience in Transaction Fee Mechanism Design}.
\newblock In {\em Proceedings of the 25th ACM Conference on Economics and Computation}, pages 1045--1073, 2024.

\bibitem{chung2023foundations}
Hao Chung and Elaine Shi.
\newblock {Foundations of Transaction Fee Mechanism Design}.
\newblock In {\em Proceedings of the 2023 Annual ACM-SIAM Symposium on Discrete Algorithms (SODA)}, pages 3856--3899. SIAM, 2023.

\bibitem{shapley_banzhaf_1}
Xiaotie Deng and Christos~H. Papadimitriou.
\newblock On the complexity of cooperative solution concepts.
\newblock {\em Mathematics of Operations Research}, 19(2):257--266, 1994.
\newblock \href {https://doi.org/10.1287/moor.19.2.257} {\path{doi:10.1287/moor.19.2.257}}.

\bibitem{diamandis2024toward}
Theo Diamandis, Tarun Chitra, and 0xShitTrader.
\newblock Toward multidimensional solana fees.
\newblock {\em Umbra Research}, 2024.
\newblock URL: \url{https://www.umbraresearch.xyz/writings/toward-multidimensional-solana-fees}.

\bibitem{diamandis2023designing}
Theo Diamandis, Alex Evans, Tarun Chitra, and Guillermo Angeris.
\newblock {Designing Multidimensional Blockchain Fee Markets}.
\newblock In {\em 5th Conference on Advances in Financial Technologies (AFT 2023)}. Schloss Dagstuhl-Leibniz-Zentrum f{\"u}r Informatik, 2023.

\bibitem{dickerson2017adding}
Thomas Dickerson, Paul Gazzillo, Maurice Herlihy, and Eric Koskinen.
\newblock {Adding Concurrency to Smart Contracts}.
\newblock In {\em Proceedings of the ACM Symposium on Principles of Distributed Computing}, pages 303--312, 2017.

\bibitem{eclipse2024local}
{Eclipse Labs}.
\newblock Local fee markets are necessary to scale ethereum.
\newblock {\em Eclipse}, 2024.
\newblock URL: \url{https://www.eclipse.xyz/articles/local-fee-markets-are-necessary-to-scale-ethereum}.

\bibitem{shapley_banzhaf_2}
Piotr Faliszewski and Lane Hemaspaandra.
\newblock The complexity of power-index comparison.
\newblock {\em Theoretical Computer Science}, 410(1):101--107, 2009.
\newblock URL: \url{https://www.sciencedirect.com/science/article/pii/S030439750800710X}, \href {https://doi.org/10.1016/j.tcs.2008.09.034} {\path{doi:10.1016/j.tcs.2008.09.034}}.

\bibitem{ferreira2021dynamic}
Matheus~VX Ferreira, Daniel~J Moroz, David~C Parkes, and Mitchell Stern.
\newblock {Dynamic Posted-Price mMchanisms for the Blockchain Transaction-Fee Market}.
\newblock In {\em Proceedings of the 3rd ACM Conference on Advances in Financial Technologies}, pages 86--99, 2021.

\bibitem{gafni2022greedy}
Yotam Gafni and Aviv Yaish.
\newblock {Greedy Transaction Fee Mechanisms for (Non-) Myopic Miners}.
\newblock {\em arXiv preprint arXiv:2210.07793}, 5, 2022.

\bibitem{gafni2024barriers}
Yotam Gafni and Aviv Yaish.
\newblock {Barriers to Collusion-Resistant Transaction Fee Mechanisms}.
\newblock In {\em Proceedings of the 25th ACM Conference on Economics and Computation}, pages 1074--1096, 2024.

\bibitem{gafni2024discrete}
Yotam Gafni and Aviv Yaish.
\newblock {Discrete and Bayesian Transaction Fee Mechanisms}.
\newblock In {\em The International Conference on Mathematical Research for Blockchain Economy}, pages 145--171. Springer, 2024.

\bibitem{gelashvili2022block}
Rati Gelashvili, Alexander Spiegelman, Zhuolun Xiang, George Danezis, Zekun Li, Yu~Xia, Runtian Zhou, and Dahlia Malkhi.
\newblock {Block-STM: Scaling Blockchain Execution by Turning Ordering Curse to a Performance Blessing}.
\newblock {\em arXiv preprint arXiv:2203.06871}, 2022.

\bibitem{heimbach2023defi}
Lioba Heimbach, Quentin Kniep, Yann Vonlanthen, and Roger Wattenhofer.
\newblock Defi and nfts hinder blockchain scalability.
\newblock In {\em International Conference on Financial Cryptography and Data Security}, pages 291--309. Springer, 2023.

\bibitem{heimbach2023dissecting}
Lioba Heimbach, Quentin Kniep, Yann Vonlanthen, Roger Wattenhofer, and Patrick Z{\"u}st.
\newblock {Dissecting the EIP-2930 Optional Access Lists}.
\newblock {\em arXiv preprint arXiv:2312.06574}, 2023.

\bibitem{hougaard2023farsighted}
Jens~Leth Hougaard and Mohsen Pourpouneh.
\newblock {Farsighted Miners Under Transaction Fee Mechanism EIP-1559}.
\newblock In {\em 2023 IEEE International Conference on Blockchain and Cryptocurrency (ICBC)}, pages 1--9. IEEE, 2023.

\bibitem{inproceedings}
Vatche Isahagian, Raymond Sweha, Azer Bestavros, and Jonathan Appavoo.
\newblock Cloudpack: Exploiting workload flexibility through rational pricing.
\newblock In {\em Proceedings of the 13th ACM/IFIP/USENIX Middleware Conference}, 12 2012.

\bibitem{keyneom2024local}
keyneom.
\newblock Local fee markets in ethereum.
\newblock {\em Ethereum Research}, 2024.
\newblock URL: \url{https://ethresear.ch/t/local-fee-markets-in-ethereum/20754}.

\bibitem{sui_object_model}
Mysten Labs.
\newblock Sui object model documentation, 2025.
\newblock Accessed: 2025-01-22.
\newblock URL: \url{https://docs.sui.io/concepts/object-model}.

\bibitem{solana_transactions}
Solana Labs.
\newblock Solana transactions documentation, 2025.
\newblock Accessed: 2025-01-22.
\newblock URL: \url{https://solana.com/de/docs/core/transactions}.

\bibitem{lavi2022redesigning}
Ron Lavi, Or~Sattath, and Aviv Zohar.
\newblock {Redesigning Bitcoin’s Fee Market}.
\newblock {\em ACM Transactions on Economics and Computation}, 10(1):1--31, 2022.

\bibitem{leonardos2021dynamical}
Stefanos Leonardos, Barnab{\'e} Monnot, Dani{\"e}l Reijsbergen, Efstratios Skoulakis, and Georgios Piliouras.
\newblock {Dynamical Analysis of the EIP-1559 Ethereum Fee Market}.
\newblock In {\em Proceedings of the 3rd ACM Conference on Advances in Financial Technologies}, pages 114--126, 2021.

\bibitem{leonardos2023optimality}
Stefanos Leonardos, Dani{\"e}l Reijsbergen, Barnab{\'e} Monnot, and Georgios Piliouras.
\newblock {Optimality Despite Chaos in Fee Markets}.
\newblock In {\em International Conference on Financial Cryptography and Data Security}, pages 346--362. Springer, 2023.

\bibitem{lostin2025truth}
Lostin.
\newblock The truth about solana local fee markets.
\newblock {\em Helius}, 2025.
\newblock URL: \url{https://www.helius.dev/blog/solana-local-fee-markets}.

\bibitem{shapley_banzhaf_3}
Tomomi Matsui and Yasuko Matsui.
\newblock A survey of algorithms for calculating power indices of weighted majority games.
\newblock {\em Journal of the Operations Research Society of Japan}, 43:71--86, Nov 2000.
\newblock \href {https://doi.org/10.15807/jorsj.43.71} {\path{doi:10.15807/jorsj.43.71}}.

\bibitem{shapley_banzhaf_4}
Yasuko Matsui and Tomomi Matsui.
\newblock Np-completeness for calculating power indices of weighted majority games.
\newblock {\em Theoretical Computer Science}, 263(1):305--310, 2001.
\newblock Combinatorics and Computer Science.
\newblock URL: \url{https://www.sciencedirect.com/science/article/pii/S0304397500002516}, \href {https://doi.org/10.1016/S0304-3975(00)00251-6} {\path{doi:10.1016/S0304-3975(00)00251-6}}.

\bibitem{monad}
{Monad Labs}.
\newblock Monad, 2023.
\newblock Accessed: 2025-01-23.
\newblock URL: \url{https://www.monad.xyz}.

\bibitem{mueller2025}
Sebastian Mueller.
\newblock X post by sebastian mueller.
\newblock \url{https://x.com/NaitsabesMue/status/1862519048069959893}, 2025.
\newblock Accessed: 2025-02-08.

\bibitem{murgia2021theory}
Maurizio Murgia, Letterio Galletta, and Massimo Bartoletti.
\newblock {A Theory of Transaction Parallelism in Blockchains}.
\newblock {\em Logical Methods in Computer Science}, 17, 2021.

\bibitem{sui}
{Mysten Labs}.
\newblock Sui, 2023.
\newblock Accessed: 2025-01-23.
\newblock URL: \url{https://sui.io}.

\bibitem{neiheiser2024pythia}
Ray Neiheiser, Arman Babaei, Ioannis Alexopoulos, Marios Kogias, and Eleftherios~Kokoris Kogias.
\newblock {Pythia: Supercharging Parallel Smart Contract Execution to Guide Stragglers and Full Nodes to Safety}.
\newblock \url{https://aftsib.com/papers/SIB24_paper_4.pdf}, 2024.
\newblock Accessed: 2025-01-23.

\bibitem{nisan2023serial}
Noam Nisan.
\newblock {Serial Monopoly on Blockchains}.
\newblock {\em arXiv preprint arXiv:2311.12731}, 2023.

\bibitem{Pei_2024}
Qi~Pei, Yipeng Wang, and Seunghee Shin.
\newblock Litmus: Fair pricing for serverless computing.
\newblock In {\em Proceedings of the 29th ACM International Conference on Architectural Support for Programming Languages and Operating Systems, Volume 4}, ASPLOS ’24, page 155–169. ACM, April 2024.
\newblock URL: \url{http://dx.doi.org/10.1145/3622781.3674181}, \href {https://doi.org/10.1145/3622781.3674181} {\path{doi:10.1145/3622781.3674181}}.

\bibitem{penna2024serial}
Paolo Penna and Manvir Schneider.
\newblock {Serial Monopoly on Blockchains with Quasi-patient Users}.
\newblock {\em arXiv preprint arXiv:2405.17334}, 2024.

\bibitem{pirlea2021practical}
George P{\^\i}rlea, Amrit Kumar, and Ilya Sergey.
\newblock {Practical Smart Contract Sharding with Ownership and Commutativity Analysis}.
\newblock In {\em Proceedings of the 42nd ACM SIGPLAN International Conference on Programming Language Design and Implementation}, pages 1327--1341, 2021.

\bibitem{shapley_banzhaf_5}
K.~Prasad and J.~S. Kelly.
\newblock Np-completeness of some problems concerning voting games.
\newblock {\em International Journal of Game Theory}, 19(1):1--9, Mar 1990.
\newblock \href {https://doi.org/10.1007/BF01753703} {\path{doi:10.1007/BF01753703}}.

\bibitem{reijsbergen2021transaction}
Dani{\"e}l Reijsbergen, Shyam Sridhar, Barnab{\'e} Monnot, Stefanos Leonardos, Stratis Skoulakis, and Georgios Piliouras.
\newblock {Transaction Fees on a Honeymoon: Ethereum's EIP-1559 One Month Later}.
\newblock In {\em 2021 IEEE International Conference on Blockchain (Blockchain)}, pages 196--204. IEEE, 2021.

\bibitem{roughgarden2020transaction}
Tim Roughgarden.
\newblock {Transaction Fee Mechanism Design for the Ethereum Blockchain: An Economic Analysis of EIP-1559}.
\newblock {\em arXiv preprint arXiv:2012.00854}, 2020.

\bibitem{roughgarden2024transaction}
Tim Roughgarden.
\newblock {Transaction Fee Mechanism Design}.
\newblock {\em Journal of the ACM}, 71(4):1--25, 2024.

\bibitem{saraph2019empirical}
Vikram Saraph and Maurice Herlihy.
\newblock {An Empirical Study of Speculative Concurrency in Ethereum Smart Contracts}.
\newblock {\em arXiv preprint arXiv:1901.01376}, 2019.

\bibitem{sei_protocol}
{Sei Labs}.
\newblock Sei protocol, 2023.
\newblock Accessed: 2025-01-23.
\newblock URL: \url{https://www.sei.io}.

\bibitem{sergey2017concurrent}
Ilya Sergey and Aquinas Hobor.
\newblock {A Concurrent Perspective on Smart Contracts}.
\newblock In {\em International Conference on Financial Cryptography and Data Security}, pages 478--493. Springer, 2017.

\bibitem{shi2022can}
Elaine Shi, Hao Chung, and Ke~Wu.
\newblock {What can Cryptography do for Decentralized Mechanism Design}.
\newblock {\em arXiv preprint arXiv:2209.14462}, 2022.

\bibitem{solana}
{Solana Labs}.
\newblock Solana, 2023.
\newblock Accessed: 2025-01-23.
\newblock URL: \url{https://solana.com}.

\bibitem{wang2024mechanism}
Wenhao Wang, Lulu Zhou, Aviv Yaish, Fan Zhang, Ben Fisch, and Benjamin Livshits.
\newblock {Mechanism Design for ZK-Rollup Prover Markets}.
\newblock {\em arXiv preprint arXiv:2404.06495}, 2024.

\bibitem{wu2023maximizing}
Ke~Wu, Elaine Shi, and Hao Chung.
\newblock {Maximizing Miner Revenue in Transaction Fee Mechanism Design}.
\newblock {\em arXiv preprint arXiv:2302.12895}, 2023.

\bibitem{10.1145/2479942.2479944}
Hong Xu and Baochun Li.
\newblock A study of pricing for cloud resources.
\newblock {\em SIGMETRICS Perform. Eval. Rev.}, 40(4):3–12, April 2013.
\newblock \href {https://doi.org/10.1145/2479942.2479944} {\path{doi:10.1145/2479942.2479944}}.

\bibitem{yao2020incentive}
Andrew Chi-Chih Yao.
\newblock {An Incentive Analysis of some {B}itcoin Fee Designs}.
\newblock In {\em Proceedings of the 47th International Colloquium on Automata, Languages, and Programming (ICALP)}, 2020.

\bibitem{zhang2018enabling}
An~Zhang and Kunlong Zhang.
\newblock {Enabling Concurrency on Smart Contracts Using Multiversion Ordering}.
\newblock In {\em Asia-Pacific Web (APWeb) and Web-Age Information Management (WAIM) Joint International Conference on Web and Big Data}, pages 425--439. Springer, 2018.

\end{thebibliography}

\appendix

\section{Proofs Omitted From \cref{sect:gcm-prop}}
\lemmaproponetwothree*
\begin{proof} \cref{prop:resource_monotonicity,prop:time_monotonicity} are special cases of \cref{prop:resource_time_monotonicity}, corresponding to the scenarios where $\txtime_1 = \txtime_2$ and $\txkeys_1 = \txkeys_2$, respectively. Therefore, the $(\Rightarrow)$ direction holds. To prove the $(\Leftarrow)$ direction, assume \cref{prop:resource_monotonicity,prop:time_monotonicity} hold and consider a set of transactions $\txset$ and two transactions $\tx_1 \simeq (\txtime_1, \txkeys_1)$ and $\tx_2 \simeq (\txtime_2, \txkeys_2)$, both not in $T$, such that $\txtime_1 \leq \txtime_2$ and $\txkeys_1 \subseteq \txkeys_2$. Let $\tx_3$ be another transaction not in $T$ such that $\tx_3 \simeq (\txtime_1, \txkeys_2)$. 
By first applying \cref{prop:resource_monotonicity} and then \cref{prop:time_monotonicity}, we obtain:
\[
\gas_{T \cup \{\tx_1\}}(\tx_1) \leq \gas_{T \cup \{\tx_3\}}(\tx_3) \leq \gas_{T \cup \{\tx_2\}}(\tx_2).
\]
This establishes the conclusion.
\end{proof}

\lemmaproponetwothreeequiv*
\begin{proof} The strict variants of \cref{prop:resource_monotonicity,prop:time_monotonicity} are special cases of the strict variant of \cref{prop:resource_time_monotonicity}, corresponding to the scenarios where $\txtime_1 = \txtime_2$ and $\txkeys_1 = \txkeys_2$, respectively. Therefore, the $(\Rightarrow)$ direction holds. To prove the $(\Leftarrow)$ direction, assume \cref{prop:resource_monotonicity,prop:time_monotonicity} hold strictly and consider a set of transactions $\txset$ and two transactions $\tx_1 \simeq (\txtime_1, \txkeys_1)$ and $\tx_2 \simeq (\txtime_2, \txkeys_2)$, both not in $T$, such that $\txtime_1 \leq \txtime_2$ and $\txkeys_1 \subseteq \txkeys_2$, and additionally $\tx_1 \not\simeq \tx_2$ (i.e., at least one of the previous holds strictly). Let $\tx_3$ be another transaction not in $T$ such that $\tx_3 \simeq (\txtime_1, \txkeys_2)$. By first applying \cref{prop:resource_monotonicity} and then \cref{prop:time_monotonicity} (their \emph{non-strict} versions), we obtain:
\begin{equation}\label{eq:ineq-chain}    
\gas_{T \cup \{\tx_1\}}(\tx_1) \leq \gas_{T \cup \{\tx_3\}}(\tx_3) \leq \gas_{T \cup \{\tx_2\}}(\tx_2).
\end{equation}

From this, we get that $\gas_{T \cup {\tx_1}}(\tx_1) \leq \gas_{T \cup {\tx_2}}(\tx_2)$, so it remains to rule out the equality case. Assume for a contradiction that $\gas_{T \cup \{\tx_1\}}(\tx_1) = \gas_{T \cup \{\tx_2\}}(\tx_2)$, from which the two inequalities hold with equality in \cref{eq:ineq-chain}. Because \cref{prop:resource_monotonicity,prop:time_monotonicity} hold strictly, this can only be the case if $\tx_1 \simeq \tx_3 \simeq \tx_2$, contradicting the assumption that $\tx_1 \not\simeq \tx_2$.
\end{proof}

\section{Proofs Omitted From \cref{sect:our-gcms}}

\banzhafnotefficient*
\begin{proof} Consider the set of transactions $T=\{\tx_1,\tx_2,\tx_3\}$ where $\tx_1\simeq(1,\{k_1\})$, $\tx_2\simeq(1,\{k_1\})$, and $\tx_3\simeq(1,\{k_2\})$. Then, for any number of threads $n \geq 2$ and the optimal scheduler, we have:
\[
v(\{\tx_1\})=v(\{\tx_2\})=v(\{\tx_3\})=v(\{\tx_1,\tx_3\})=v(\{\tx_2,\tx_3\})=1,\quad
v(\{\tx_1,\tx_2\}) = v(T) = 2.
\]
A short calculation then shows that:
\[
\gas_T^B(\tx_1)=\gas_T^B(\tx_2)=\frac{1+1+0+1}{4}=\frac{3}{4},\quad
\gas_T^B(\tx_3) = \frac{1 + 0 + 0 + 0}{4} = \frac{1}{4}.
\]
Thus, \(\sum_{\tx\in T}\gas_T^B(\tx)=\frac{3}{4}+\frac{3}{4}+\frac{1}{4}=\frac{7}{4}\neq 2 = v(T)\), violating Efficiency.
\end{proof}

\section{Proofs Omitted From \cref{sect:gcm-analysis}}

\subsection{Transaction Bundling (\cref{prop:tx_bundling})}
\lemmashapleynotxbundling*
\begin{proof} Consider the set of transactions \(\txset=\{\tx_4\}\) and three other transactions $\tx_1, \tx_2, \tx_3$ such that $\tx_3$ is the concatenation of $\tx_1$ and $\tx_2$, where  $\tx_1\simeq(1,\{k_2\})$, $\tx_2\simeq(1,\{k_2\})$, $\tx_3\simeq(2,\{k_2\})$, and $\tx_4\simeq(1,\{k_1\})$. Then, for any number of threads $n \geq 2$ and the optimal scheduler, we have:
\begin{align*}
    \gas_{\txset \cup \{\tx_1, \tx_2\}}^S(\tx_1)=\gas_{\txset \cup \{\tx_1, \tx_2\}}^S(\tx_2)&=\frac{0+1+1+1+1+1}{6} \\
    \gas_{\txset \cup \{\tx_3\}}^S(\tx_3)&=\frac{2+1}{2}.
\end{align*}

Thus, \(\gas_{\txset \cup \{\tx_1, \tx_2\}}^S(\tx_1)+\gas_{\txset \cup \{\tx_1, \tx_2\}}^S(\tx_2)=\frac{10}{6}>\frac{3}{2}= \gas_{\txset \cup \{\tx_3\}}^S(\tx_3)\), which violates Transaction Bundling.
\end{proof}

\lemmabanzhafhastxbundl*
\begin{proof}
    Consider a set of transactions $\txset$ and three transactions $\tx_1, \tx_2, \tx_3 \notin T$ such that $\tx_3$ is the concatenation of $\tx_1$ and $\tx_2$. Then, for any scheduler satisfying property \ref{sched:monot-bundle}, we have:
    \begin{align*}
        \gas^B_{\txset\cup \{\tx_1, \tx_2\}}(\tx_1) &+ \gas^B_{\txset\cup \{\tx_1, \tx_2\}}(\tx_2)
        \\&= \frac{1}{2^{|\txset\cup \{\tx_1, \tx_2\}| - 1}}\sum_{S\subseteq (\txset\cup\{\tx_1, \tx_2\})\setminus\{\tx_1\}} \left[ v ( S\cup\{ \tx_1\} ) - v(S) \right] \\&+ \frac{1}{2^{|\txset\cup \{\tx_1, \tx_2\}| - 1}}\sum_{S\subseteq (\txset\cup\{\tx_1, \tx_2\})\setminus\{\tx_2\}} \left[ v ( S\cup\{ \tx_2\} ) - v(S) \right]
        \\&=\frac{1}{2^{|\txset\cup \{\tx_1, \tx_2\}| - 1}}\sum_{S\subseteq \txset} \left[ v ( S\cup\{ \tx_1\} ) - v(S) + v ( S\cup\{ \tx_1, \tx_2\} ) - v(S\cup\{\tx_2\}) \right] \\&+ \frac{1}{2^{|\txset\cup \{\tx_1, \tx_2\}| - 1}}\sum_{S\subseteq \txset} \left[ v ( S\cup\{ \tx_2\} ) - v(S) + v ( S\cup\{ \tx_1,\tx_2\} ) - v(S\cup\{\tx_1\}) \right]
        \\&=\frac{2}{2^{|\txset\cup \{\tx_1, \tx_2\}| - 1}}\sum_{S\subseteq \txset} \left[v ( S\cup\{ \tx_1, \tx_2\} ) - v(S)\right]
        \\&\leq\frac{1}{2^{|\txset\cup \{\tx_3\}| - 1}}\sum_{S\subseteq \txset} \left[v ( S\cup\{ \tx_3\} ) - v(S)\right] \quad = \quad \gas^B_{\txset\cup \{\tx_3\}}(\tx_3). \qedhere
    \end{align*}
\end{proof}

\lemmatpmhasbundling*
\begin{proof}
    Consider a set of transactions $\txset$ and three transactions $\tx_1, \tx_2, \tx_3 \notin T$ such that $\tx_3$ is the concatenation of $\tx_1$ and $\tx_2$. 
    Then, for any scheduler satisfying property \ref{sched:monot-bundle}, we have:
    \begin{align*}
        \gas^\textit{TPM}_{\txset\cup \{\tx_1, \tx_2\}}(\tx_1) &+ \gas^\textit{TPM}_{\txset\cup \{\tx_1, \tx_2\}}(\tx_2)
        \\&=\frac{t(\tx_1)}{\sum_{\tx' \in \txset\cup\{\tx_1,\tx_2\}}t(\tx')} \cdot v(\txset\cup\{\tx_1,\tx_2\})+\frac{t(\tx_2)}{\sum_{\tx' \in \txset\cup\{\tx_1,\tx_2\}}t(\tx')} \cdot v(\txset\cup\{\tx_1,\tx_2\})
        \\&=\frac{t(\tx_1)+t(\tx_2)}{\sum_{\tx' \in \txset\cup\{\tx_1,\tx_2\}}t(\tx')} \cdot v(\txset\cup\{\tx_1,\tx_2\})
        \\&=\frac{t(\tx_3)}{\sum_{\tx' \in \txset\cup\{\tx_3\}}t(\tx')} \cdot v(\txset\cup\{\tx_1,\tx_2\})
        \\&\leq\frac{t(\tx_3)}{\sum_{\tx' \in \txset\cup\{\tx_3\}}t(\tx')} \cdot v(\txset\cup\{\tx_3\}) \quad = \quad \gas^\textit{TPM}_{\txset\cup \{\tx_3\}}(\tx_3). \qedhere
    \end{align*}
\end{proof}

\lemmaesmnobundling*
\begin{proof} Consider the set of transactions \(\txset=\{\tx_4\}\) and three transactions $\tx_1, \tx_2, \tx_3 \notin \txset$ such that $\tx_3$ is the concatenation of $\tx_1$ and $\tx_2$, where $\tx_1\simeq(1,\{k_1\})$, $\tx_2\simeq(1,\{k_1\})$, $\tx_3\simeq(2,\{k_1\})$, and $\tx_4\simeq(1,\{k_1\})$. Then, for any number of threads $n \geq 2$ and the optimal scheduler, we have:
\[
    \gas_{\txset \cup \{\tx_1, \tx_2\}}^\textit{ESM}(\tx_1)=\gas_{\txset \cup \{\tx_1, \tx_2\}}^\textit{ESM}(\tx_2) = \frac{3}{3}, \quad \gas_{\txset \cup \{\tx_3\}}^\textit{ESM}(\tx_3) = \frac{3}{2}.
\]

Thus, \(\gas_{\txset \cup \{\tx_1, \tx_2\}}^\textit{ESM}(\tx_1)+\gas_{\txset \cup \{\tx_1, \tx_2\}}^\textit{ESM}(\tx_2)=2>\frac{3}{2}= \gas_{\txset \cup \{\tx_3\}}^\textit{ESM}(\tx_3)\), which violates Transaction Bundling.
\end{proof}

\lemmaxsmhassricttxbundling*
\begin{proof}
    Consider a set of transactions $\txset$ and three transactions $\tx_1, \tx_2, \tx_3 \notin T$ such that $\tx_3$ is the concatenation of $\tx_1$ and $\tx_2$. 
    Then, for any scheduler satisfying property \ref{sched:monot-bundle}, we have:
    \begin{align*}
        \gas^\textit{XSM}_{\txset\cup \{\tx_1, \tx_2\}}(\tx_1) &+ \gas^\textit{XSM}_{\txset\cup \{\tx_1, \tx_2\}}(\tx_2)
        \\&=\frac{v(\txset\cup \{\tx_1, \tx_2\})}{3^{|\txset\cup \{\tx_1, \tx_2\}|}}+\frac{v(\txset\cup \{\tx_1, \tx_2\})}{3^{|\txset\cup \{\tx_1, \tx_2\}|}}
        \\&=\frac{2}{3}\cdot\frac{v(\txset\cup \{\tx_1, \tx_2\})}{3^{|\txset\cup\{\tx_3\}|}}
        \\&<\frac{v(\txset\cup \{\tx_1, \tx_2\})}{3^{|\txset\cup\{\tx_3\}|}}
        \\&\leq\frac{v(\txset\cup \{\tx_3\})}{3^{|\txset\cup\{\tx_3\}|}} =\gas^\textit{XSM}_{\txset\cup \{\tx_3\}}(\tx_3). \qedhere
    \end{align*}
\end{proof}

\subsection{Set Inclusion (\cref{prop:set_inclusion})}

\lemmashapleynotsetincl*
\begin{proof} Consider the transaction sets \(\txset=\{\tx_1\}\), \(\txset_1=\{\tx_2, \tx_3\}\), and \(\txset_2=\{\tx_2,\tx_3,\tx_4\}\) where $\tx_1\simeq(1,\{k_1\})$, $\tx_2\simeq(1,\{k_2\})$, $\tx_3\simeq(1,\{k_2\})$, and $\tx_4\simeq(1,\{k_1\})$.
Then, for any number of threads $n \geq 2$ and the optimal scheduler, we have:
\[
    \gas_{\txset \cup \txset_1}^S(\txset_1)=\frac{5}{6}+\frac{5}{6}, \quad \quad \gas_{\txset \cup \txset_2}^S(\txset_2)=\frac{12}{24}+\frac{12}{24}+\frac{12}{24}.
\]
Thus, \(\gas_{\txset \cup \txset_1}^S(\txset_1)=\frac{10}{6}>\frac{36}{24}=\gas_{\txset \cup \txset_2}^S(\txset_2)\), which violates Set Inclusion.
\end{proof}

\lemmabanzhafnosetincl*
\begin{proof} Consider the transaction sets \(\txset=\emptyset\), \(\txset_1=\{\tx_1, \tx_2\}\), and \(\txset_2=\{\tx_1,\tx_2,\tx_3\}\) where $\tx_1\simeq(1,\{k_1\})$, $\tx_2\simeq(1,\{k_1\})$, and $\tx_3\simeq(1,\{k_2\})$.
Then, for any number of threads $n \geq 2$ and the optimal scheduler, we have:
\[
    \gas_{\txset \cup \txset_1}^B(\txset_1)=\frac{1+1}{2}+\frac{1+1}{2}, \quad \quad \gas_{\txset \cup \txset_2}^B(\txset_2)=\frac{1+1+0+1}{4}+\frac{1+1+0+1}{4}+\frac{1+0+0+0}{4}.
\]
Thus, \(\gas_{\txset \cup \txset_1}^B(\txset_1)=2>\frac{7}{4}=\gas_{\txset \cup \txset_2}^B(\txset_2)\), which violates Set Inclusion.
\end{proof}

\lemmatpmhassetincl*
\begin{proof}
    Consider a set of transactions  $\txset$ and two sets of transactions $\txset_1 \subseteq \txset_2$, disjoint from $\txset$. 
    Then, for any scheduler satisfying property \ref{sched:monot-t}, we have:
    \begin{align*}
        \gas^\textit{TPM}_{\txset\cup \txset_1}(\txset_1)&=\frac{\sum_{\tx\in\txset_1}t(\tx)}{\sum_{\tx' \in \txset\cup\txset_1}t(\tx')} \cdot v(\txset\cup\txset_1)
        \\&\leq \frac{\sum_{\tx\in\txset_1}t(\tx)+\sum_{\tx\in\txset_2\backslash\txset_1}t(\tx)}{\sum_{\tx' \in \txset\cup\txset_1}t(\tx')+\sum_{\tx'\in\txset_2\backslash\txset_1}t(\tx')} \cdot v(\txset\cup\txset_1)
        \\&= \frac{\sum_{\tx\in\txset_2}t(\tx)}{\sum_{\tx' \in \txset\cup\txset_2}t(\tx')} \cdot v(\txset\cup\txset_1)
        \\&\leq \frac{\sum_{\tx\in\txset_2}t(\tx)}{\sum_{\tx' \in \txset\cup\txset_2}t(\tx')} \cdot v(\txset\cup\txset_2) \quad = \quad \gas^\textit{TPM}_{\txset\cup \txset_2}(\txset_2). \qedhere
    \end{align*}
\end{proof}

\lemmaesmhassetincl*
\begin{proof}
    Consider a set of transactions  $\txset$ and two sets of transactions $\txset_1 \subseteq \txset_2$, disjoint from $\txset$. 
    Then, for any scheduler satisfying property \ref{sched:monot-t}, we have:
    \begin{align*}
        \gas^\textit{ESM}_{\txset\cup \txset_1}(\txset_1)&=\sum_{\tx\in\txset_1}\frac{v(\txset\cup\txset_1)}{|\txset\cup\txset_1|}
        \\&=\frac{|\txset_1|}{|\txset\cup\txset_1|}\cdot v(\txset\cup\txset_1)
        \\&\leq\frac{|\txset_1|+|T_2\backslash\txset_1|}{|\txset\cup\txset_1|+|T_2\backslash\txset_1|}\cdot v(\txset\cup\txset_1)
        \\&=\frac{|\txset_2|}{|\txset\cup\txset_2|}\cdot v(\txset\cup\txset_1)
        \\&\leq\frac{|\txset_2|}{|\txset\cup\txset_2|}\cdot v(\txset\cup\txset_2) \quad = \quad \gas^\textit{ESM}_{\txset\cup \txset_2}(\txset_2). \qedhere
    \end{align*}
\end{proof}

\lemmaxsmnosetinclusion*
\begin{proof} Consider the transaction sets \(\txset=\emptyset\), \(\txset_1=\{\tx_1\}\), and \(\txset_2=\{\tx_1,\tx_2\}\), where $\tx_1\simeq(1,\{k_1\})$ and $\tx_2\simeq(1,\{k_2\})$.
Then, for any number of threads $n \geq 2$ and the optimal scheduler, we have:
\[
    \gas_{\txset \cup \txset_1}^\textit{XSM}(\txset_1)=\frac{1}{3^1}, \quad \quad \gas_{\txset \cup \txset_2}^\textit{XSM}(\txset_2)=\frac{1}{3^2}+\frac{1}{3^2}.
\]
Thus, \(\gas_{\txset \cup \txset_1}^\textit{XSM}(\txset_1)=\frac{1}{3}>\frac{2}{9}=\gas_{\txset \cup \txset_2}^\textit{XSM}(\txset_2)\), which violates Set Inclusion.
\end{proof}

\subsection{Storage Key-Time Monotonicity (\cref{prop:resource_time_monotonicity}) and Strict Time Monotonicity (\cref{prop:time_monotonicity})}

\lemmatpmrestimesensandtimesens*
\begin{proof} 
Consider a set of transactions $\txset$ and two transactions $\tx_1 \simeq (\txtime_1, \txkeys_1)$ and $\tx_2 \simeq (\txtime_2, \txkeys_2)$, both not in $T$,
    such that $\txtime_1 \leq \txtime_2$ and $\txkeys_1 \subseteq \txkeys_2$. Then, for any scheduler satisfying property \ref{sched:monot-t-R}:
    \begin{align*}
        \gas^\textit{TPM}_{\txset\cup \{\tx_1\}}(\tx_1)&=\frac{t_1}{\sum_{\tx' \in \txset\cup\{\tx_1\}}t(\tx')} \cdot v(\txset\cup\{\tx_1\})
        \\&=\frac{t_1}{\sum_{\tx' \in \txset}t(\tx')+t_1} \cdot v(\txset\cup\{\tx_1\})
        \\&\leq\frac{t_2}{\sum_{\tx' \in \txset}t(\tx')+t_2} \cdot v(\txset\cup\{\tx_1\})
        \\&\leq\frac{t_2}{\sum_{\tx' \in \txset}t(\tx')+t_2} \cdot v(\txset\cup\{\tx_2\}) \quad = \quad \gas^\textit{TPM}_{\txset\cup \{\tx_2\}}(\tx_2).
    \end{align*}
    This establishes Storage Key-Time Monotonicity.

    To also get strict Time Monotonicity, assume  $t_1 < t_2$ and $K_1 = K_2$ in the above proof. Note that, when
    \(\sum_{\tx' \in \txset} t(\tx') > 0\), the function \(\frac{t}{\sum_{\tx'\in \txset} t(\tx') + t}\) is strictly increasing in \(t\), guaranteeing that the first inequality above is strict, so
    \(\gas_{\txset\cup \{\tx_1\}}^\textit{TPM}(\tx_1) < \gas_{\txset\cup \{\tx_2\}}^\textit{TPM}(\tx_2)\), as required. In the degenerate case where \(\sum_{\tx' \in \txset} t(\tx') = 0\), the fraction becomes 1 in both cases, but this can only happen when $T = \varnothing$, in which case the value \(v\) increases strictly because $v(\{\tx_1\}) = t_1 < t_2 = v(\{\tx_2\})$. Overall, the first inequality above again holds strictly, ensuring that \(\gas_{\txset\cup \{\tx_1\}}^\textit{TPM}(\tx_1) < \gas_{\txset\cup \{\tx_2\}}^\textit{TPM}(\tx_2)\) still holds, as required. Thus, the strict Time Monotonicity property is satisfied.
\end{proof}

\end{document}